\documentclass[journal,10pt]{IEEEtran}

\usepackage{cite}
\usepackage{amsmath,amssymb,amsfonts}
\usepackage{algorithmic}
\usepackage{graphicx}
\usepackage{textcomp}
\usepackage{xcolor}
\usepackage{graphicx}
\usepackage{caption}
\usepackage[utf8]{inputenc}
\DeclareSymbolFont{matha}{OML}{txmi}{m}{it}
\DeclareMathSymbol{\varv}{\mathord}{matha}{118}
\usepackage{stackengine}
\def\delequal{\mathrel{\ensurestackMath{\stackon[1pt]{=}{\scriptstyle\Delta}}}}
\usepackage[english]{babel}
\usepackage{amsthm}
\usepackage{algorithm}
\usepackage{setspace}
\usepackage{epstopdf}
\usepackage{lipsum}
\usepackage[figurename=Fig.]{caption}
\newtheorem{theorem}{Theorem}
\newtheorem{proposition}[theorem]{Proposition}
\newtheorem{remark}{Remark}
\usepackage{soul}
\usepackage{comment}

\newcommand{\yellow}{\textcolor{black}}
\newcommand{\blue}{\textcolor{black}}
\usepackage{tabularx,booktabs}
\usepackage{adjustbox}
\newcolumntype{C}{>{\centering\arraybackslash}X}
\usepackage{hyperref}
\hypersetup{
    bookmarks=true,         
    unicode=false,          
    pdftoolbar=true,        
    pdfmenubar=true,        
    pdffitwindow=false,     
    pdfstartview={FitH},    
    pdftitle={My title},    
    pdfauthor={Author},     
    pdfsubject={Subject},   
    pdfcreator={Creator},   
    pdfproducer={Producer}, 
    pdfkeywords={keyword1, key2, key3}, 
    pdfnewwindow=true,      
    colorlinks=true,       
    linkcolor=blue,          
    citecolor=blue,
    filecolor=cyan,         
    urlcolor=magenta        
}

\begin{document}
\title{Low Complexity Turbo SIC-MMSE Detection for Orthogonal Time Frequency Space Modulation}

\author{Qi Li,~\IEEEmembership{Student Member,~IEEE,}
        Jinhong Yuan,~\IEEEmembership{Fellow,~IEEE,}
        Min Qiu,~\IEEEmembership{Member,~IEEE,}
        Shuangyang Li,~\IEEEmembership{Member,~IEEE,}
        and Yixuan Xie,~\IEEEmembership{Member,~IEEE}

\thanks{
The work of Qi Li, Jinhong Yuan, Min Qiu, and Yixuan Xie was supported in part by the Australian Research Council (ARC) Discovery Project under Grant DP220103596, and in part by the ARC Linkage Project under Grant LP200301482. The work of Shuangyang Li is supported in part by the European Union’s Horizon 2020 Research and Innovation Program under MSCA Grant No. 101105732 – DDComRad.\par
This work has been presented in part at the 2022 IEEE International Conference on Communications Workshops \cite{Li2022IterativeModulation}.\par
Qi Li, Jinhong Yuan, Min Qiu, and Yixuan Xie are with the School of Electrical Engineering and Telecommunications, The University of New South Wales, Sydney, Australia (e-mail: oliver.li1@unsw.edu.au, j.yuan@unsw.edu.au, min.qiu@unsw.edu.au, yixuan.xie@unsw.edu.au).\par
Shuangyang Li is with Technical University of Berlin, Berlin, Germany (e-mail: shuangyang.li@tu-berlin.de).

}%
}

\maketitle

\begin{abstract}
Recently, orthogonal time frequency space (OTFS) modulation has garnered considerable attention due to its robustness against doubly-selective wireless channels. In this paper, we propose a low-complexity iterative successive interference cancellation based minimum mean squared error (SIC-MMSE) detection algorithm for zero-padded OTFS (ZP-OTFS) modulation. In the proposed algorithm, signals are detected based on layers processed by multiple SIC-MMSE linear filters for each sub-channel, with interference on the targeted signal layer being successively canceled either by hard or soft information. To reduce the complexity of computing individual layer filter coefficients, we also propose a novel filter coefficients recycling approach in place of generating the exact form of MMSE filter weights. Moreover, we design a joint detection and decoding algorithm for ZP-OTFS to enhance error performance. Compared to the conventional SIC-MMSE detection, our proposed algorithms outperform other linear detectors, e.g., maximal ratio combining (MRC), for ZP-OTFS with up to 3 dB gain while maintaining comparable computation complexity.

\end{abstract}


\IEEEpeerreviewmaketitle

\section{Introduction}
\IEEEPARstart{W}{ith} the deployment of 5G wireless network expanding globally, wireless communication is increasingly being used in high-mobility scenarios, such as high-speed railways, autonomous vehicles, unmanned aerial vehicles (UAVs), etc. The time-invariant channel assumption, which holds in conventional digital modulation, such as orthogonal frequency division multiplexing (OFDM), is no longer valid in these high-mobility scenarios due to high Doppler frequency spreads. The resulting time-variant channel leads to doubly-selectivity of the system, causing severe sub-carrier non-orthogonality in conventional OFDM modulation. To solve this problem, orthogonal time frequency space (OTFS) \cite{Hadani2017OrthogonalModulation} has been proposed as a promising technique that is more robust against the Doppler effect experienced in doubly-selective wireless channels. In OTFS, information symbols are modulated in the delay-Doppler (DD) domain, which distinguishes itself from the widely recognized time-frequency (TF) domain modulation.  Within a stationary region in the DD domain, the doubly-selective channel response can be modeled as \blue{quasi time-invariant}. OTFS takes advantage of the stable DD domain wireless channel by modulating symbols in this domain. In principle, the DD domain modulated symbols spread across the entire TF domain, \blue{which can be exploited to achieve} full time-frequency diversity crucial for signal processing on the receiver \cite{LR3,Wei2021OrthogonalWaveform}. However, OTFS requires tailored detection algorithms, as \blue{it is difficult to directly} adopt a simple single-tap equalizer on the receiver side due to strong inter-symbol interference (ISI) experienced in the effective DD domain channel. Thus, low-complexity yet powerful detectors are necessary to achieve the desired error performance of the modulation.\par
There have been numerous detection algorithms proposed for OTFS modulation. For example, the conventional non-linear detector, maximum a posteriori (MAP), achieves optimal performance, but its complexity is impractically high \cite{LR1,LR2} that grows exponentially with the number of paths. To address this, several message passing (MP)-based detection variants have been proposed. In \cite{LR3}, the authors presented a message passing algorithm (MPA) that \blue{performs in \yellow{DD} domain which can be presented by a sparse factor graph and applies the Gaussian approximation to the interference} to reduce computation complexity. However, it was later discovered in \cite{LR4} that MPA may converge to a locally optimal point in the loopy factor graph. To overcome this issue, the authors proposed a variational Bayes (VB) approach that guarantees convergence and yields significant performance improvement. Although these MPA detector variants significantly reduce complexity from MAP, they sacrifice near-optimal performance.\par
Analyses of linear detectors, such as zero forcing (ZF) and minimum mean squared error (MMSE) for OTFS modulation have been conducted in \cite{LR24}. In \cite{LR25}, the complexity of linear MMSE (LMMSE) matrix inversion was reduced by utilizing LU matrix decomposition. Successive interference cancellation (SIC) and parallel interference cancellation (PIC) for LMMSE filtering were applied by the authors in \cite{LR27} and \cite{LR28}, respectively. In \cite{LR30}, a maximal ratio combining (MRC) \blue{based RAKE receiver} was presented for a zero-padding OTFS (ZP-OTFS) system. This work revealed \blue{a simple vector form of ZP-OTFS input-output relation in the time domain for which lower-complexity detection becomes possible.} Cross-domain detection by means of MMSE filtering was originally investigated by authors in \cite{Li2021CrossModulation}. Inspired by these works, we aim to strengthen the error performance from MRC by proposing a cross-domain signal detection with a comparable complexity for ZP-OTFS. An improved version of the general MMSE, SIC-MMSE, has been well investigated for multiple-input-multiple-output (MIMO) OFDM systems, such as the vertical Bell Labs layered space-time (V-BLAST) \cite{LR33}, \cite{LR34}, \cite{LR35}, as discussed in works \cite{Lee2005NewSystems, Wang2007SoftStudy}. Based on these findings, we propose a high-performance and low-complexity detection algorithm and integrate it with cross-domain signal processing for ZP-OTFS. In our system model, the time domain channel \blue{enables}  low-complexity SIC-MMSE detection. We discover that the MMSE filter weights calculated from the previous signal can be reused for the subsequent signal detection, reducing the computation complexity significantly. We demonstrate that this recycling approach does not degrade the error performance when based on a preset mean squared error threshold. \par
It is \blue{well-known} that turbo codes \blue{can} approach the Shannon capacity \cite{LR36}. \blue{Because of this, the turbo principle has been extensively applied to} the inter-symbol-interference (ISI) channel with soft-input-soft-output joint equalization and decoding \cite{LR37}, \cite{LR38}, \cite{LR39}, \cite{LR40}. These pioneering works provide the foundation for our tailored turbo detection and decoding algorithm, specifically designed for ZP-OTFS.\par {{In this paper, we introduce a novel SIC-MMSE detection algorithm and a turbo receiver for ZP-OTFS system. We introduce two novel algorithms: the hard/soft SIC-MMSE and the approximate SIC-MMSE. Our research demonstrates that these innovative algorithms deliver substantial performance enhancements when compared to other linear detection methods.} \blue{Nevertheless, the main contributions of this work can be} summarized as follows.

\begin{itemize}
\item \blue{By exploiting the simple vector form of ZP-OTFS input-output relation \cite{LR30}, we propose a low-complexity iterative MMSE cross-domain signal detection algorithm with SIC, where detection and interference cancellation are performed on a symbol-by-symbol basis.} {Unlike traditional SIC-MMSE \cite{Chatterjee2021NonorthogonalTransmission} and MMSE \cite{Li2021CrossModulation} techniques proposed for OTFS, our SIC-MMSE relies on cross-domain iterative processing, addressing each signal layer's individual sub-channel. This approach results in a notable complexity reduction to $\mathcal{O}((M-l_{max})Nl_{max}^3)$ per iteration (where $M$ and $N$ are the numbers of subcarriers and time slots, respectively.
$l_{max}$ is the largest delay spread), while the classical SIC-MMSE operates on the entire channel, incurring a high complexity of $\mathcal{O}((M-l_{max})^3N^3)$ per iteration.} \blue{Specifically, we propose soft and hard SIC-MMSE detection algorithms. We show that both algorithms achieve a similar bit error rate (BER) for low-order modulations, e.g., 4-ary quadrature amplitude modulation (QAM) while the soft SIC-MMSE detection outperforms its hard counterpart for high-order modulations.}

\item \blue{To further reduce the computational complexity, we propose an approximate SIC-MMSE method, where we recycle MMSE filter weights for a fixed number of to-be-detected signals and reuse them for the detection of the subsequent signals. {The approximate SIC-MMSE can reduce the previous complexity $\mathcal{O}((M-l_{max})Nl_{max}^3)$ further by a factor of hundreds of times, bringing it down to $\mathcal{O}((M-l_{max})Nl_{max}^3\times \frac{1}{\Delta m})$, where $\Delta m$ is the number of approximated subchannels. A more detailed complexity analysis will be provided in Sec. IV.} We also derive an error tolerance parameter for the MMSE filter output $\Delta\beta$ and establish its relation to the number of approximated sub-channels allowed $\Delta m$. With a proper choice of $\Delta\beta$, the proposed approximate SIC-MMSE detector can achieve BER very close to the original SIC-MMSE detector. In addition, we show that for a given $\Delta\beta$, $\Delta m$ is inversely proportional to the Doppler shift experienced in the channel.}

\item \blue{For the proposed SIC-MMSE detectors, we analyze the signal-to-interference-plus-noise ratio (SINR) for the detection outputs and discuss its computation complexity.} We also precisely characterize the error recursion across iterations and derive upper and lower bounds for the cases where the residual interference remains full and completely canceled, respectively. \blue{We demonstrate} that the proposed approximate SIC-MMSE detector can approach the \blue{MSE} lower bound.

\item \blue{Finally, we introduce a joint detection and decoding turbo receiver for a coded ZP-OTFS system. Our numerical results validate the superior performance of the proposed algorithm over the benchmark OTFS detectors in the literature.}
\end{itemize}

The paper is organized as follows: In Section II, we provide a brief overview of modulation and demodulation for ZP-OTFS, as well as the input-output relation and effective channel matrix in the time domain. In Section III, we introduce the iterative cross-domain SIC-MMSE signal detection with both hard interference cancellation and soft interference cancellation. Section IV presents an approximate SIC-MMSE algorithm that achieves significant complexity reduction by reusing MMSE filter coefficients while maintaining the optimal SIC-MMSE performance from Section III. In Section V, we construct a soft-input-soft-output (SISO) turbo receiver for coded ZP-OTFS systems, where significant iterative gains are observed compared to both the uncoded system and that in \cite{LR30}. We also analyze the performance of our proposed detector and derive the state evolution to accurately characterize the error recursions. In Section VI, we provide extensive simulation results with various system parameters and compare them with other existing detectors for ZP-OTFS. Finally, Section VII concludes our findings and suggests future research directions.

\subsection*{Notations:}
We will use the following notations throughout this paper: $x$, $\bf x$, and $\bf X$ represent scalar, vector, and matrix. \blue{The operations $(.)^T$, $(.)^{\dag}$, and $(.)^H$ represent transpose, conjugate transpose, and Hermitian transpose, respectively.} $\bf 0$, ${\bf I}_{N}$, ${\bf F}_{N}$ and ${\bf F}_{N}^{\dag}$ are zero matrices, identity matrix with order $N$, $N$-point normalized discrete Fourier transform (DFT) matrix and $N$-point normalized inverse discrete Fourier transform (IDFT) matrix, respectively. Let $\circledast$, $\otimes$, and $\circ$ denote circular convolution, Kronecker product, and Hadamard product. $\text{vec}(\bf X)$ and $\text{vec}_{N,M}^{-1}(\bf x)$ represent column-wise vectorization of matrix $\bf X$ and matrix formed by folding a vector $\bf x$ into a $N\times M$ matrix, respectively. The set of $N\times M$ dimensional matrices with complex entries are denoted by $\mathbb {C}^{N \times M}$. \blue{For a channel matrix \yellow{${\bf H}\in\mathbb {C}^{MN \times MN}$}}, ${\bf H}_{n,m}$ denotes the $m$-th sub-channel matrix at the $n$-th signal block \blue{of ${\bf H}$} and ${\bf H}_{n,m}[:,l]$ denotes the $l$-th column vector in matrix ${\bf H}_{n,m}$.


\section{OTFS System Models}
Throughout this paper, we adopt the ZP-OTFS modulation system as presented in \cite{LR30}. The ZP utilized in OTFS is comparable to the CP or ZP added in OFDM, which simplifies the equalization process at the receiver \cite{LR30}. Additionally, the ZP incorporated into the delay-Doppler grid can serve as the guard band for pilot-based channel estimation \cite{Raviteja2019EmbeddedChannels}. In the following section, we provide a detailed explanation of the modulation and demodulation process at the transmitter and receiver, as well as the input-output relation in the ZP-OTFS system.
\subsection{Transmitter}
Let $x[k,j]$, $k\in [1,M]$, $j\in [1,N]$, denote the $(k,j)$-th modulated  $\mathcal{Q}$-QAM data symbols that are arranged on the two-dimensional DD domain grid ${\bf X}_{\rm DD} \in \mathbb {C}^{M \times N}$, where $N$ and $M$ are the number of time slots/Doppler bins and sub-carriers/delay bins per OTFS frame, respectively. The zero paddings are added at the last $l_{max}$ rows of ${\bf X}_{\rm DD}$, where $l_{max}$ is the largest delay index of the channel response. Moving forward, the symbols of a constellation set $\mathcal{Q}$ in the DD domain will be first transformed into the TF domain by using the inverse symplectic fast Fourier transform (ISFFT) given as \cite{Hadani2017OrthogonalModulation}
\begin{align} X[n,m]=\frac{1}{\sqrt {NM}} \sum _{k=0}^{N-1}{\sum _{j=0}^{M-1}{x[k,j] e^{j2\pi \left({\frac {nk}{N}-\frac {mj}{M}}\right)}}} ,\label{eq1} \end{align} where $X[n,m]$ represents the $(n,m)$-th symbol in the two-dimensional TF domain ${\bf X}_{\rm TF}$. (\ref{eq1}) can also be written using a matrix format with normalized DFT matrix as \begin{align} {\bf X}_{\rm TF}={{\bf F}} _M\cdot {{\bf X}_{\rm DD}}\cdot {{\bf F}} _N^{\dag }.\label{eq2} \end{align}
Then, the TF domain symbols are further transformed into time domain through a multicarrier (MC) modulator \cite{Hadani2017OrthogonalModulation}
\begin{align}
{s[q]= \sum _{n=0}^{N-1}{\sum _{m=0}^{M-1}}{X[n,m]g_{\rm tx}\bigg(\frac {qT}{M}-nT \bigg) e^{j2\pi m \Delta f(\frac {qT}{M}-nT)}}},
\end{align} where $g_{\rm tx}(t)$ is the transmit shaping pulse with sampling rate $\frac {T}{M}$, $\Delta f$ and $T$ are the subcarrier spacing and symbol period of the MC modulation, respectively. The time domain discrete signals in the matrix format can then be expressed as \begin{align} {{\bf s}} = \text{vec}({{\bf G}} _{\rm tx}\cdot ({{\bf F}} _M^{\dag }\cdot {\bf X}_{\rm TF})) , \end{align} where ${{\bf G}} _{\rm tx}$ is the matrix representation of the pulse shaping filter $g_{\rm tx}(t)$ at the transmitter. In this paper, only rectangular pulse\footnote{\blue{For OTFS, rectangular pulse shapes can introduce high out-of-band emission (OOBE) and cause ISI when a practical band pass filter is applied on receiver \cite{Shen2022ErrorReceivers}. To address this issue, \cite{Lin2022OrthogonalModulation} introduced orthogonal delay-Doppler division multiplexing modulation, which uses a Nyquist pulse train to achieve low OOBE as well as local orthogonality on the DD plane.}} will be considered on both transmitter and receiver sides as ideal pulse shape cannot be realized in practice \cite{Raviteja2019PracticalOTFS}. This leads to an identity pulse shaping matrix ${{\bf G}}_{\rm tx}$, i.e., ${\bf G}_{\rm tx}={\bf I}_M$. The time domain discrete signals are generated by an $M$-point FFT and an $N$-point IFFT of ${\bf X}_{\rm DD}$ from ISFFT, followed by an $M$-point IFFT from Heisenberg transform. Therefore, the time domain signals can be equivalently generated by an $N$-point IFFT for each row of ${\bf X}_{\rm DD}$, which is also known as the inverse discrete Zak transform (IDZT). Therefore, the time domain discrete OTFS signals at the transmitter side can be written in a vectorized form \cite{LR30}, i.e.,
\begin{align}
&{\bf s}=\text{vec}({\bf X}_{\rm DD}\cdot {\bf F}_N^{\dag }).  \label{eq5}
\end{align}

\subsection{Channel}
For a high-mobility wireless channel, the channel can be represented by a linear time-variant (LTV) system, a.k.a. the doubly-selective channel \cite{Bello1963CharacterizationChannels}. A physically meaningful representation of the LTV wireless channel is based on the time delays and Doppler frequency shifts \cite{Bello1963CharacterizationChannels}. Let $P$ denote the number of \blue{resolvable} paths in the channel, and the DD domain channel can be represented as \cite{Bello1963CharacterizationChannels} \begin{align} h(\tau, \nu) = \sum _{i=1}^{P} h_i \delta (\tau -\tau _i) \delta (\nu -\nu _i), \label{eq26}\end{align}
where $\tau _i$, $\nu _i$ and $h_i$ are the delay, Doppler shift, and attenuation factor for the $i$-th path, respectively \cite{Bello1963CharacterizationChannels}. Let $\tau_{i} =\ell_i \frac{T}{M}$, where $\ell_i$ is the delay tap for the $i$-th path and $\frac{T}{M}$ is the sampling interval or delay resolution, and $\boldsymbol{\kappa}_{\ell_i}$ represent all the Doppler corresponding to the delay $\ell_i$. Note that $\ell_i$ may not be an integer. The continuous impulse response of the channel in the delay-time domain is related to (\ref{eq26}) by
\begin{align} h(\tau,t)&= \int_\nu h(\tau, \nu){ e}^{j2\pi \nu (t-\tau)}d\nu =\sum^P_{i=1}\sum _{\kappa\in\boldsymbol{\kappa}_{\ell_i}} h_i{e}^{j2\pi \kappa \frac{\Delta f}{N} (t-\ell_i \frac{T}{M})}. \label{eq261}
\end{align}
Now let us denote $h(\ell_i,t)=\sum _{\kappa\in\boldsymbol{\kappa}_{\ell_i}} h_i{e}^{j2\pi \kappa \frac{\Delta f}{N} (t-\ell_i \frac{T}{M})}$. Then we have $h(\tau,t) = \sum^P_{i=1}h(\ell_i,t)$. The discrete baseband equivalent channel seen by the receiver is the sinc reconstructed impulse response having discrete integer multiples $n$ of sampling intervals $T_s=\frac{T}{M}$ on time $t$ \cite{LR30}. The discrete baseband equivalent delay-time channel model can then be generally expressed as
\begin{align} h_e(l,n)=&\sum^P_{i=1} h(\ell_i,t)|_{t=n\frac{T}{M}} {\rm sinc}(l-\ell_i)\nonumber \\ =&\sum^P_{i=1} \sum _{\kappa\in\boldsymbol{\kappa}_{\ell_i}} h_i{e}^{j2\pi \kappa \frac{(n-\ell_i)}{NM}} {\rm sinc}(l-\ell_i), \label{eq28}
\end{align} where ${\rm sinc}(x)={\rm sin}(\pi x)/{(\pi x)}$. (\ref{eq28}) considers all off-grid delay $\ell_i$ and Doppler $\kappa$ in the system model.
\subsection{Receiver}
At the receiver side, the received signal $r(t)$ can be modeled as a convolution with the channel response $h(\tau,t)$ given by \cite{LR30}\begin{align} r(t)= \int _{0}^{\tau_{max}} h(\tau,t)s(t-\tau)\, d\tau \label{eq230} \end{align} Then, $r(t)$ is first sent into a multicarrier demodulator with a matched filtering to obtain $Y[n,m]$ \cite{Hadani2017OrthogonalModulation} \begin{align} Y[n,m]=&\int g_{rx}^{*}(t-nT) r(t) e^{-j2 \pi m\Delta f(t-nT)} \mathrm {d}t \label{eq231} \end{align}
Based on $Y[n,m]$, the symplectic fast Fourier transform (SFFT) is performed to convert TF signals into DD signals, yielding \begin{align} y[k,j] = \frac {1}{\sqrt {NM}}\sum _{n=0}^{N-1} \sum _{m=0}^{M-1} Y[n,m] e^{-j2\pi \left({{\frac{nk }{ N}}-{\frac{mj }{ M}}}\right)}. \label{eq233}\end{align}
\subsection{Input-Output Relations for ZP-OTFS in Time Domain}
The structure of the time domain channel lacks resolved Doppler shifts, which presents an opportunity to simplify signal detection in this domain, especially when the channel has a complex Doppler response. The linear input-output relation for all signals in the time domain can be expressed in a holistic manner
\begin{align}
    {\bf r}={\bf H}\cdot {\bf s}+{\bf z},
\end{align}
where $\bf r$ is the time domain received signal vector: ${\bf r} = [{\bf r}_{0}^T,...,{\bf r}^T_{N-1}]^T$ $\in \mathbb {C}^{MN\times 1}$ while ${\bf r}_{n}$, $n \in\{ 0,...,N-1\}$, represents the received $n$-th block signals with $M$ elements. And $\bf s$ is the time domain transmitted signals ${\bf s}=[{\bf s}_{0}^T,...,{\bf s}_{N-1}^T]^T $ $\in \mathbb {C}^{MN\times 1}$, where each ${\bf s}_{n}$ represents the $n$-th transmitted block signals also having $M$ elements. Additionally, $\bf z$ is the time domain independent and identically distributed (i.i.d) additive white Gaussian noise (AWGN) with variance $\sigma_{n}^2$. ${\bf H}$ is the time domain channel matrix that has $N$ blocks ${\bf H}_{n}$ along the main diagonal, i.e., ${\bf H}=\text {diag}\{{{\bf H}_{0}},...,{\bf H}_{N-1}\}$ $\in \mathbb {C}^{MN\times MN}$, which is shown in Fig. \ref{f:fig21}.
\begin{figure}[!t]
\includegraphics[width=\linewidth]{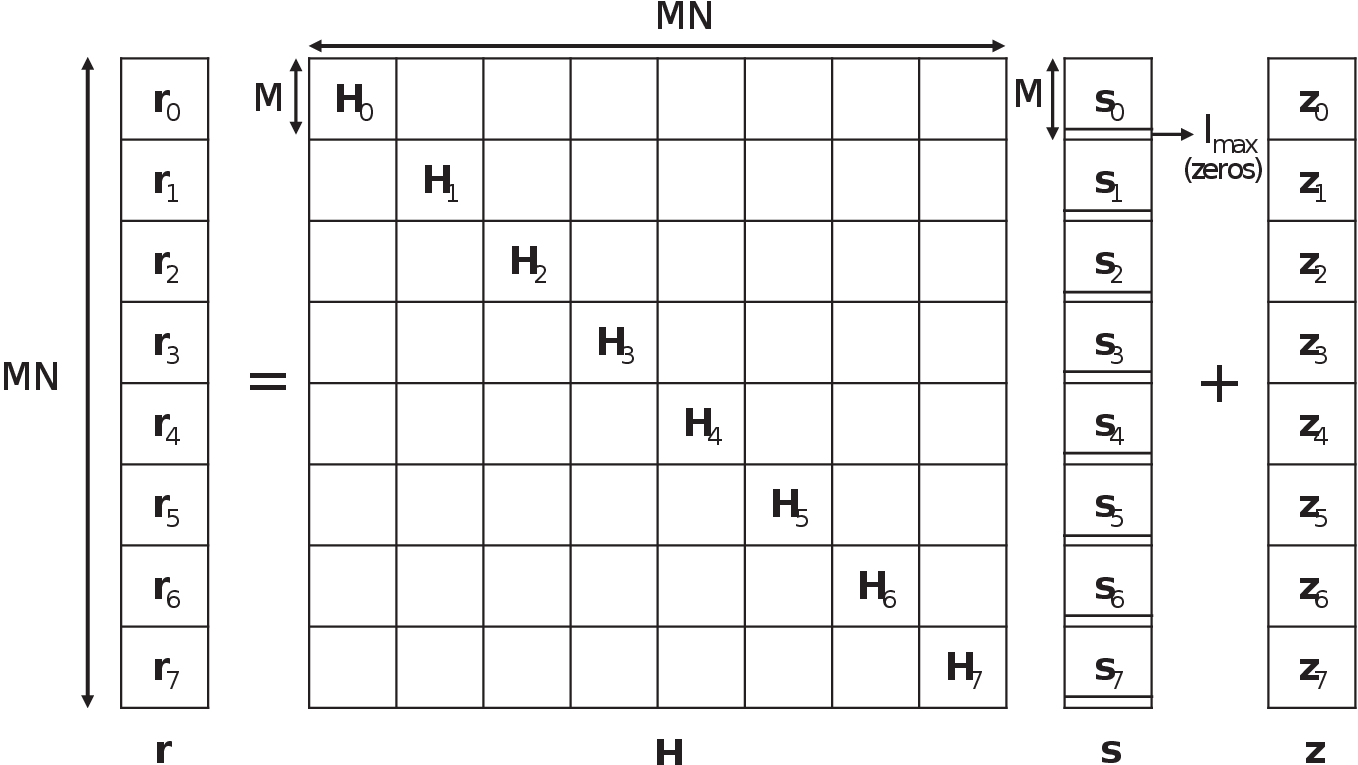}
\caption{Time domain input-output relation for all symbols with $N=8,M=8,l_{max}=3$.}
\label{f:fig21}
\end{figure}
The zero paddings added into the DD grid partition the time domain signals into $N$ orthogonal blocks ${\bf s}_{n}$, where each signal block ${\bf s}_{n}$ does not interfere with any other blocks ${\bf s}_{n'}, \forall n'\neq n$. This allows parallel processing for all $N$ blocks simultaneously \cite{LR30}. For this reason, any detection method that is applied in one of the blocks ${\bf s}_{n}$ can be concurrently executed for the rest of the blocks due to the independence. Now let us look closely at the input-output relation at one particular transmitted signal block ${\bf s}_{n}$.\par
The input-output relation at one particular transmitted signal block ${\bf s}_{n}$ can be expressed as
\begin{align}
{{\bf r}_{n}}={{\bf H}_{n}}\cdot {{\bf s}_{n}}+{{\bf z}_{n}},
\end{align}
where ${\bf r}_{n} = [{r}_{n,0},...,{r}_{n,(M-1)}]^T \in \mathbb {C}^{M\times 1}$ is the $n$-th received signal block that contains $M$ elements and each ${r}_{n,m}$ represents the $m$-th element in the $n$-th block. ${\bf s}_{n} = [{s}_{n,0},...,{s}_{n,(M-1)}]^T \in \mathbb {C}^{M\times 1}$ is the $n$-th transmitted signal block and each ${s}_{n,m}$ represents the $m$-th element in the $n$-th \blue{transmitted} block. ${{\bf H}_{n}}\in \mathbb {C}^{M\times M}$ is the channel matrix for the $n$-th transmitted signal block and the elements in ${{\bf H}_{n}}$ will be provided later, and ${\bf z}_{n}$ is the noise vector with variance $\sigma_{n}^2$.
Note that the last few elements in ${\bf s}_{n}$ are zeros due to the insertion of zero paddings in DD domain. As for the composition of ${\bf H}_{n}$, the position of non-zero elements $h_{m,m-\tilde{l}}$, $\forall \tilde{l} \in l$ where $l$ is the set of delay taps for all $P$ paths as defined in (\ref{eq28}), are determined by the delay index $\tilde{l}$. The rest of the elements $h_{m,m-\tilde{l}}$, $\forall \tilde{l} \notin l$, are zeros that are left blank in Fig. \ref{f:fig22}. The derivation of non-zero elements in ${\bf H}_{n}$ are given by \cite{LR30} \begin{align} {h}_{m,m-\tilde{l}}=\tilde{\mathbf { \boldsymbol {\nu } }}_{m,\tilde{l}}[n], \label{eq235} \end{align} where $\tilde{\boldsymbol {\nu }}_{m,\tilde{l}}[n] = \sum _{\kappa\in\boldsymbol{\kappa}_{\tilde{l}}} h_i\alpha^{\kappa (m-\tilde{l})}{e}^{\frac{j2\pi \kappa n}{N}}$, and we define {$\alpha \triangleq e^{\frac{j2\pi}{NM}}$}. Moving forward, we refer to each element ${s}_{n,m}$ in ${\bf s}_{n}$ as a ‘signal layer’ (or the $m$-th layer). It is evident that each signal layer corresponds to only a fraction of the entire channel ${\bf H}_{n}$, and we shall refer to this portion of the channel as a ‘sub-channel’ of layer ${s}_{n,m}$. These sub-channels are \blue{highlighted in different colors} in Fig. \ref{f:fig22}. Notably, the largest size of the sub-channel is only $(l_{max}+1)$ $\times$ $(2l_{max}+1)$. \blue{This means that the complexity of computation of MMSE filtering weights calculation can be made lower, which will be }discussed in the upcoming section. Furthermore, due to the insertion of zero paddings, the first layer at one symbol block is not interfered with by the layer at the last symbol block. Therefore, the detection process can be initiated directly from the first layer without requiring interference cancellation.
\begin{figure}[!t]
\includegraphics[width=\linewidth]{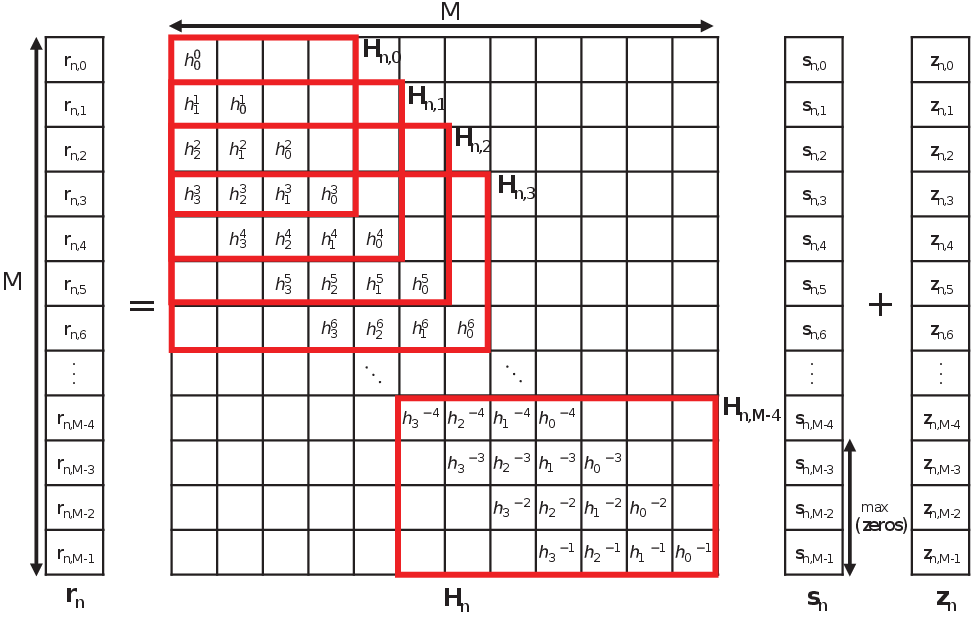}
\caption{{Time domain input-output relation for the $n$-th block with $l_{max}=3$.}}
\label{f:fig22}
\end{figure}

\section{Iterative SIC-MMSE with Hard and Soft Estimates}
In this section, we investigate the SIC-MMSE detection for the ZP-OTFS system. The addition of zero paddings in the DD domain renders the $N$ blocks in the time domain orthogonal, allowing for parallel processing among the $N$ blocks. Within each block, the transmitted symbols are detected based on layers, as explained in the previous section. The symbol layers, which are interference to the target layer, are successively canceled using their statistical information from the received signal vector. Then, the received signal vector undergoes MMSE filtering to maximize the signal-to-interference-plus-noise-ratio (SINR) of the target transmit signal. Our proposed methodologies investigate both hard and soft interference cancellations. Although hard interference cancellation is simpler, soft estimates can achieve a significant performance improvement, particularly for higher-order modulations.
\subsection{SIC-MMSE Detection in the Time Domain with Hard Interference Cancellation}
Without loss of generality, we only focus on detecting the $(M-l_{max})$ time domain signals at one particular block ${\bf s}_{n}$, and the rest of the blocks follow the same methodology. Let $s_{n,m}$ denote the $m$-th transmit symbol at the $n$-th block and ${\bar{\bf r}}_{n,m}$ denote the receive signal vector that contains \blue{$s_{n,m}$} due to the multipath. For instance, in Fig. \ref{f:fig22}, the receive signal vector that corresponds to the signal layer $s_{n,0}$ is ${{\bar{\bf r}}_{n,0}} = [{r_{n,0},r_{n,1},r_{n,2}, r_{n,3}}]^T$, and the sub-channel for the signal layer $s_{n,0}$ is the matrix ${\bf H}_{n,0}$. In general, the received symbol vector $\bar{{\bf r}}_{n,m} \in \mathbb {C}^{(l_{max}+1)\times 1}$ can be expressed as
\begin{align}
{{\bar{\bf r}}_{n,m}} =& \sum_{j=0}^{l'+l_{max}}{\bf H}_{n,m}[:,j]{s}_{n,(m'+j)}+{\bf z}_{n,m}, \nonumber\\ &0 \le m \le M-l_{max}, \label{eq300}
\end{align}
where $m'=\max\{m-l_{max},0\}$, $l'=\min\{m,l_{max}\}$, and ${\bf H}_{n,m}[:,j]$ is the $j$-th column in ${\bf H}_{n,m}$, {${\bf z}_{n,m}$ is the time domain Gaussian noise vector having the same length as ${{\bar{\bf r}}_{n,m}}$}.
The time domain hard interference cancellation in the $i$-th iteration can then be written as
\begin{align}
{\hat{\bf r}^{(i)}_{n,m}} =& {\bar{\bf r}}_{n,m}-\sum_{j=0}^{l'-1}{\bf H}_{n,m}[:,j]\bar{s}_{n,(m'+j)}^{(i)}\nonumber\\&-\sum_{k=l'+1}^{l'+l_{max}}{\bf H}_{n,m}[:,k]\bar{s}_{n,(m'+k)}^{(i-1)}, \label{eq310}
\end{align}
where ${\bar s}^{(i)}$ and ${\bar s}^{(i-1)}$ are the hard estimates of time domain interference variables detected in the $i$-th iteration and $(i-1)$-th iteration respectively and {${\bar s}^{(i)}$ are initiated as zeros in the first iteration}. The estimated output signal ${\hat{s}}_{n,m}$ after going through the MMSE filtering is then
\begin{align}
{\hat{s}}^{(i)}_{n,m} &= {\bf w}_{n,m}^{(i)}{\hat{\bf r}^{(i)}_{n,m}}\nonumber\\&= \mu^{(i)}_{n,m}s^{(i)}_{n,m}+{\hat {z}^{(i)}_{n,m}}, \label{eq320}
\end{align}
where $\mu^{(i)}_{n,m} = {\bf w}^{(i)}_{n,m}{\bf H}_{n,m}[:,l']$, ${\hat {z}^{(i)}_{n,m}}$ is the MMSE-suppressed interference, and ${\bf w}_{n,m}^{(i)}$ are the MMSE filter coefficients and it is calculated as
\begin{align}
{\bf w}^{(i)}_{n,m} = {\bf H}^{H}_{n,m}[:,l']({\bf H}_{n,m}{\bf V}_{n,m}{\bf H}_{n,m}^{H}+\sigma_{n}^2{{\bf I}_{l_{\max}+1}})^{-1}, \label{eq330}
\end{align}
where ${\bf H}_{n,m}$ is the sub-channel matrix, ${\bf H}_{n,m}[:,l']$ is the $l'$-th column of ${\bf H}_{n,m}$ and $l'$ is defined after (\ref{eq300}), ${\bf V}_{n,m}$ is the covariance matrix. Since hard decisions are considered, it is initialized as {${\bf V}_{n,m} = $ {diag}$\{0,...,0,E_s,...,E_s\}$ during the first iteration and ${\bf V}_{n,m} = $ {diag}$\{0,...,0,E_s,0,...,0\}$} from the second iteration and after, where $E_{s}$ is the symbol energy.
We then normalize ${\hat{s}}^{(i)}_{n,m}$ in (\ref{eq320}) by
\begin{align}
{\hat{s}}^{(i)}_{n,m} = \frac{{\hat{s}}^{(i)}_{n,m}}{\mu^{(i)}_{n,m}}=s^{(i)}_{n,m}+\frac{{\hat {z}^{(i)}_{n,m}}}{\mu^{(i)}_{n,m}}. \label{eq39}
\end{align}

After detecting the $m$-th layer at the $n$-th block $s_{n,m}$, the $m$-th layer of the rest of the $N-1$ blocks will follow the same procedure in (\ref{eq310}) and (\ref{eq320}). When all the blocks have been processed for the $m$-th signal layer, we will then have a length $N$ estimated signal vector ${\hat{\bf s}}^{(i)}_{m} = [{\hat s}_{0,m}^{(i)},{\hat s}_{1,m}^{(i)},...,{\hat s}^{(i)}_{(N-1),m}]$ from the time domain SIC-MMSE detection. Then, they are transformed into DD domain via the $N$-point FFT
\begin{align}
{\hat{\bf y}}_{m}= {{\bf F}_{N}}{\hat{\bf s}^{(i)}_{m}}.\label{eq33}
\end{align}
Based on the DD domain observation ${\hat{\bf y}}_{m}$, the hard decisions are made by the maximum likelihood (ML) criterion
\begin{align}
{\hat{{x}}}_{m}[n]=\arg \min _{a\in \mathcal {Q}} \left|{a-{\hat{y}}_{m}}[n]\right|,\label{eq34}
\end{align}
where $a$ is \blue{a} symbol from \blue{the constellation set} $\mathcal{Q}$. \blue{The hard estimate symbol vector from \eqref{eq34}} ${\hat{\bf x}}_{m}$ will be transformed back into time domain through an $N$-point IFFT
\begin{align}
{\bf \bar{s}}^{(i+1)}_{m}= {{\bf F}_{N}^{\dag }}\hat{{\bf x}}_{m}.\label{eq35}
\end{align}
From here, the newly acquired hard estimates in the time domain ${\bf \bar{s}}^{(i+1)}_{m}$ will be used as initial estimates for interference cancellation in (\ref{eq310}), during the next iteration of detection.

\subsection{Soft Interference Cancellation}
The soft interference cancellation for SIC-MMSE essentially takes advantage of the constellation constraint in the DD domain, and continuously improves the statistical information of the estimated transmit symbols in DD domain via iterations. Given the received signal vector ${\bf r}_{n,m}$ for the target signal layer $s_{n,m}$ and the sub-channel ${\bf H}_{n,m}$, the time domain soft interference cancellation in the $i$-th iteration is similar to (\ref{eq310}) and it is given by
{\begin{align}
{\hat{\bf r}^{(i)}_{n,m}} =& \bar{{\bf r}}_{n,m}-\sum_{j=0}^{l'-1}{\bf H}_{n,m}[:,j]\tilde{s}_{n,(m'+j)}^{(i)}\\&-\sum_{k=l'+1}^{l'+l_{max}}{\bf H}_{n,m}[:,k]\tilde{s}_{n,(m'+k)}^{(i-1)} \nonumber\\
=&\sum^{l'+l_{max}}_{j=0}{\bf H}{[:,j]}s_{n,m'+j}+{\bf z}_{n,m}-\sum_{j=0}^{l'-1}{\bf H}_{n,m}[:,j]\tilde{s}_{n,(m'+j)}^{(i)}\nonumber\\&-\sum_{k=l'+1}^{l'+l_{max}}{\bf H}_{n,m}[:,k]\tilde{s}_{n,(m'+k)}^{(i-1)}, \label{eq36}
\end{align}}
where $\tilde{s}_{n,(m'+j)}^{(i)}$ and $\tilde{s}_{n,(m'+k)}^{(i-1)}$ are now the soft estimates of time domain interference detected in current $i$-th iteration and last $(i-1)$-th iteration, respectively. Breaking $\sum^{l'+l_{max}}_{j=0}{\bf H}{[:,j]}s_{n,m'+j}$ into $\sum^{l'-1}_{j=0}{\bf H}_{n,m}{[:,j]}s_{n,m'+j}+\sum^{l'+l_{max}}_{j=l'+1}{\bf H}_{n,m}{[:,j]}s_{n,m'+j}+{\bf H}_{n,m}{[:,l']}s_{n,m}$, (\ref{eq36}) can be rewritten as
{\begin{align}
{\hat{\bf r}^{(i)}_{n,m}}=&\sum^{l'-1}_{j=0}{\bf H}_{n,m}{[:,j]}s_{n,m'+j} - \sum_{j=0}^{l'-1}{\bf H}_{n,m}[:,j]\tilde{s}_{n,(m'+j)}^{(i)}\nonumber\\&+\sum^{l'+l_{max}}_{j=l'+1}{\bf H}_{n,m}{[:,j]}s_{n,m'+j} - \sum^{l'+l_{max}}_{k=l'+1}{\bf H}_{n,m}{[:,k]}\tilde{s}_{n,(m'+k)}^{(i-1)}\nonumber\\&+{\bf H}_{n,m}{[:,l']}s_{n,m}+{\bf z}_{n,m}\nonumber\\
=&{\bf H}_{n,m}[:,l']s_{n,m}+{\hat {\bf z}^{(i)}_{n,m}},
\end{align}} where ${\hat {\bf z}^{(i)}_{n,m}}$ is the noise plus residual interference and it is thus given by
\begin{align}
{\hat {\bf z}^{(i)}_{n,m}} =& \sum\limits_{j=0}^{l'-1}{\bf H}_{n,m}[:,j]e_{n,(m'+j)}^{(i)}+\sum\limits_{k=l'+1}^{l'+l_{max}}{\bf H}_{n,m}[:,k]e_{n,(m'+k)}^{(i-1)}\nonumber\\&+{\bf z}_{n,m},
\end{align} where $e^{(i)}_{n,(m'+j)}$ and $e
^{(i-1)}_{n,(m'+k)}$ are the residual interferences left in the current iteration's and last iteration's interference cancellation process, respectively. {The residual interference is defined as $e^{(i)}_{n,m} = s_{n,m} - \tilde{s}^{(i)}_{n,m}$}. After soft interference cancellation is performed, the MMSE filter is applied to the interference canceled signal vector ${\hat{\bf r}^{(i)}_{n,m}}$. The filter coefficients are calculated \blue{in the same way as} (\ref{eq330}), \blue{except that} the covariance matrix is {${\bf V}_{n,m} = $ {diag}$\{{\mathbb V}\{e^{(i)}_{n,(m'+1)}\},...,{\mathbb V}\{e^{(i)}_{n,(m'+1'-1)}\}, {E_s}, {\mathbb V}\{e^{(i-1)}_{n,(m'+1'+1)}\},\\...,{\mathbb V}\{e^{(i-1)}_{n,(m'+1'+l_{max})}\}\}$}, in which ${\mathbb V}\{{\cdot}\}$ denotes the variance and ${\mathbb V}\{e^{(i-1)}_{n,(m'+k)}\}=E_s$ for $i=1$ and $k \in \mathbb \{l'+1,...,l'+l_{max}\}$. The details on the variance calculation will be provided later. After applying the MMSE filtering to ${\hat{\bf r}}^{(i)}_{n,m}$, the output estimate ${\hat{s}^{(i)}_{n,m}}$ can be directly written in terms of the time domain input variable $s^{(i)}_{n,m}$ as \cite{Lee2005NewSystems}

\begin{align}
{\hat{s}}^{(i)}_{n,m} =& {\bf w}^{(i)}_{n,m}{\hat{\bf r}}^{(i)}_{n,m}\nonumber\\ =&{\bf w}^{(i)}_{n,m}{\bf H}_{n,m}[:,l']s^{(i)}_{n,m}+\sum_{j=1}^{l'-1} {\bf w}^{(i)}_{n,m}{\bf H}^{H}_{n,m}[:,j]e_{n,(m'+j)}^{(i)}\nonumber\\&+ \sum\limits_{k=l'+1}^{l'+l_{max}} {\bf w}^{(i)}_{n,m}{\bf H}_{n,m}[:,k]e_{n,(m'+k)}^{(i-1)}+{\bf w}^{(i)}_{n,m}{\bf z}^{(i)}_{n,m} \nonumber\\=& \mu^{(i)}_{n,m}s^{(i)}_{n,m}+{\tilde {z}^{(i)}_{n,m}}, \label{eq38}
\end{align}
where $\mu^{(i)}_{n,m} = {\bf w}^{(i)}_{n,m}{\bf H}_{n,m}[:,l']$, ${\bf w}^{(i)}_{n,m}$ follows \eqref{eq330}, ${\tilde {z}^{(i)}_{n,m}}$ is the MMSE-suppressed noise plus residual interference and it is assumed to be Gaussian distributed \cite{Lee2005NewSystems}. Then, ${\hat{s}}^{(i)}_{n,m}$ is normalized as in (\ref{eq39}).
Moreover, the post-MMSE variance of ${\hat{s}}^{(i)}_{n,m}$ can then be simply computed by \cite{Lee2005NewSystems}
\begin{align}
\sigma_{\tilde{z}^{(i)}_{n,m}}^2 = \frac{\mu^{(i)}_{n,m}-|\mu^{(i)}_{n,m}|^2}{|\mu^{(i)}_{n,m}|^2}E_{s}. \label{eq399}
\end{align}
After performing SIC-MMSE for the $m$-th symbol in all the other $N-1$ blocks, we then have the time domain estimated symbol vector ${\hat{\bf s}}^{(i)}_{m} = [{{\hat s}}^{(i)}_{0,m},...,{{\hat s}}^{(i)}_{(N-1),m}]^T$
and the post-MMSE covariance matrix ${\tilde{\bf V}_{m}}$ with each diagonal element is obtained from (\ref{eq399}), i.e., ${\tilde{\bf V}_{m}}=$ {diag} $\{\sigma_{\tilde{z}^{(i)}_{0,m}}^2,...,\sigma_{\tilde{z}^{(i)}_{(N-1),m}}^2\}$. Then, ${\hat{\bf s}}^{(i)}_{m}$ will be transformed into DD domain via an $N$-point FFT
\begin{align}
{\hat{\bf y}}_{m}= {{\bf F}_{N}}{\hat{\bf s}^{(i)}_{m}} = {\bf x}_{m}+{{\bf F}_{N}}{\tilde{\bf{z}}_{m}},\label{eq311}
\end{align}
where ${\bf x}_{m}$ is the $m$-th symbol vector with $N$ elements and ${{\bf F}_{N}}{\tilde{\bf{z}}_{m}}$ is the equivalent noise in DD domain. The covariance matrix for ${{\bf F}_{N}}{\tilde{\bf{z}}_{m}}$ is assumed to be the same as ${\tilde{\bf V}_{m}}$ due to unitary transformation \cite{Li2021CrossModulation}. Finally, based on the observations ${\hat{\bf y}}_{m}$, the DD domain detection can be carried out in a symbol-by-symbol fashion
\begin{align}
{\bar{{x}}}_{m}[n]=\arg \max _{x_{m}[n]}{\text{Pr} }\{x_{m}[n]|\hat{y}_{m}[n]\}.\label{eq312}
\end{align}
According to the Bayes rule, the probability density function in (\ref{eq312}) is proportional to ${\text {Pr}}\{\hat{y}_{m}[n]|x_{m}[n]\}$, i.e. ${\text {Pr}}\{x_{m}[n]|\hat{y}_{m}[n]\} \propto {\text {Pr}}\{\hat{y}_{m}[n]|x_{m}[n]\}$. Thus the normalized conditional probability density function can be computed as
\begin{align}
{\text {Pr}}\{\hat{y}_{m}[n]|x_{m}[n]=a\} = \frac{\text {exp}\Bigl(\frac{|\hat{y}_{m}[n]-a|^2}{{\tilde{{\bf V}}_{m}[n,n]}}\Bigl)}{{\sum_{a\in \mathcal {Q}}}\text {exp}\Bigl(\frac{|\hat{y}_{m}[n]-a|^2}{{\tilde{{\bf V}}_{m}[n,n]}}\Bigl)}. 
\end{align}
Therefore, we can obtain the \emph{a posteriori} mean a.k.a. the soft estimate in the DD domain by
\begin{align}
{\tilde{{x}}}_{m}[n]&={\mathbb {E}}\{\hat{y}_{m}[n]|x_{m}[n]\}= \sum_{a\in \mathcal {Q}}{\text {Pr}}\{\hat{y}_{m}[n]|x_{m}[n]=a\}\times a. \label{eq312222}
\end{align}
And the \emph{a posteriori} variance
\begin{align}
{{\bf V}}_{m}[n,n]&= {\mathbb {E}\{{|x_{m}[n]-\mathbb {E}}\{\hat{y}_{m}[n]|x_{m}[n]|^2\}\}}\nonumber \\ &=\sum_{a\in \mathcal {Q}}{\text {Pr}}\{\hat{y}_{m}[n]|x_{m}[n]=a\}\times |a-{\tilde{{x}}}_{m}[n]|^2.\label{eq313}
\end{align}
After performing symbol-by-symbol detection for all the $N$ symbols in the DD domain, we will have the soft estimate signal vector ${\tilde{{\bf x}}}_{m}$ with covariance matrix ${\bf V}_{m}$. Then, ${\tilde{\bf x}}_{m}$ is transformed back to the time domain through the $N$-point IFFT, and the time domain covariance matrix is
\begin{align}
{{\bf V}^{(i)}_m}= {{\bf F}_{N}^{\dag}}{\bf V}_{m}{\bf F}_{N}.\label{eq355}
\end{align}The non-diagonal elements in ${{\bf V}^{(i)}_m}$ are not of interest due to the independence assumptions for time domain variables \cite{Li2021CrossModulation}. Finally, the new updated time domain soft estimates ${\tilde{\bf s}}^{(i)}_{m}$ along with their variances ${{\bf V}^{(i)}_m}$ will be used as the \emph{a priori} information for detecting the subsequent signal layer,  i.e., ${\tilde{s}}^{(i)}_{n,m}={\tilde{\bf s}}^{(i)}_{m}[n]$ and ${\mathbb V}\{e^{(i)}_{n,m}\} = {\bf V}^{(i)}_m[n,n]$.
\begin{remark}
{Without loss of generality, the algorithms proposed for ZP-OTFS can be extended to CP-OTFS without difficulties. In CP based systems, interferences may emerge among different blocks, rendering all sub-channel sizes to be $(2l_{max}+1)\times(2l_{max}+1)$, and this channel condition aligns with the scenario encountered by signal layers when the index of targeted signal layer $m$ exceeds ‘$l_{max}$’ in ZP model already. Consequently, the SIC-MMSE algorithm presented in this paper remains applicable and effective for other alternative systems.}
\end{remark}

\section{Low complexity approximate SIC-MMSE}
Despite the fact that the proposed cross-domain SIC-MMSE is a low-complexity algorithm, it is still required to compute $(M-l_{max})N$ MMSE matrix inversion with a total complexity of order $\mathcal{O}((M-l_{max})Nl^3_{max})$. To achieve further complexity reduction, we introduce an approximate SIC-MMSE detector in this section.

\subsection{Approximate SIC-MMSE}
The use of orthogonal $N$ signal blocks in the ZP-OTFS system enables the design of a low-complexity SIC-MMSE algorithm in the time domain, as introduced in the previous section. The iterative interference cancellation of detected layers enables the MMSE filter to better suppress other interfering signals based on its estimated detection output. As the process proceeds, residual interference power continuously decreases, leading to an increase in SINR. Each sub-channel/signal layer requires the calculation of MMSE filter weights, leading to a total of $MN$ sub-channel matrix inversions. In the following, the precise calculation of each individual sub-channel's MMSE filter weights is referred to as the \blue{‘exact MMSE'}, as opposed to the forthcoming approximate MMSE. We observe that small phase variations in adjacent sub-channels may allow for the recycling or reuse of previous filter weights when detecting new signal layers. For example, for the $m$-th signal layer at the $n$-th block, the \blue{exact} MMSE filter coefficient is denoted by ${\bf w}^{*}_{n,m}$. Without recycling, the MMSE weight calculations require to be performed by $N(M-l_{max})$ times. \blue{However}, we can use ${\bf w}^{*}_{n,m}$ as an approximation for the rest of $(\Delta m + m)$ sub-channels at the $n$-th block, where $\Delta m$ is the number of layers to be approximated following the $m$-th layer. \blue{In this way, the number of weights calculations is reduced to $\frac{N(M-l_{max})}{\Delta m}$.} In Fig. 2, for example, when detecting the third layer signal, i.e., $m = 2$, if $\Delta m$ is set as 2, then the estimation of the fourth layer signal and fifth layer signal will be approximated by using the exact MMSE weights ${\bf w}^{*}_{n,2}$ of the third layer. Interestingly, it is also found that the relation of the $m$-th and $(\Delta m + m)$-th sub-channel in the time domain is phase-rotated by the circular shifted Doppler matrix ${\boldsymbol {\nu}} =  {\text {circ}}\{\nu_{1},\nu_{2},...,\nu_{l_{max}}\} = \left[\begin{array}{cccccccc}{{ {\nu }}_{l_{max}}} & \cdots & {{{\nu }}_{2}}&{{\nu }}_{1} & {{ {\nu }}_{l_{max}}} & \cdots & {{{\nu }}_{2}}\\ {{{\nu }}_{1}} & \cdots & {{{\nu }}_{3}}&{{{\nu }}_{2}} & {{{\nu }}_{1}} & \cdots & {{{\nu }}_{3}}\\ \vdots & \ddots & \vdots&\vdots & \vdots & \ddots & \vdots \\ {{{\nu }}_{l_{max}-1}} & \cdots & {{{\nu }}_{1}}&{{{\nu }}_{l_{max}}} & {{{\nu }}_{l_{max}-1}} & \cdots & {{{\nu }}_{1}} \end{array}\right]$ $\in \mathbb {C}^{l_{max}\times (2l_{max}-1)}$ at the exponential order of $\Delta m$, wherein ${{\nu}_{i}}$ is the Doppler tap for the $i$-th path and each column of ${\boldsymbol {\nu}}$ is the circular shift of the previous column. More specifically, the $l_{max}$-th column is $[{\nu }_{1}, {\nu }_{2},...,{\nu }_{l_{max}}]^{T}$, representing the Doppler for the target signal. \blue{Columns $1,\ldots,l_{max}-1$} represent the Dopplers for the previously detected interference signals. \blue{Columns $l_{max}+1,\ldots,2l_{max}-1$} represents the Dopplers for the subsequent to-be-detected interference signals. The relation of the sub-channel ${\bf H}_{n,\Delta m + m}$ and ${\bf H}_{n,m}$ can be represented by the Hadamard (element-wise) product in terms of the Doppler matrix ${\boldsymbol {\nu}}$
\begin{align}
{\bf H}_{n,\Delta m + m} = {\bf H}_{n,m} \circ \boldsymbol ({e^{j\frac{2\pi}{MN}{\boldsymbol \nu}}})^{\Delta m}, m\geq {{l_{max}}}. \label{42}
\end{align}
\blue{It} can be observed from (\ref{42}) that the phase change caused by the Doppler shift gradually increases as $\Delta m$ increases\blue{.} \blue{This} implies that the difference of the $m$-th sub-channel and $(\Delta m + m)$-th sub-channel becomes more substantial as $\Delta m$ increases. When \blue{$\Delta m$ is large and} the approximated filter coefficients are used on the $(\Delta m + m)$-th sub-channel, the MSE is no longer minimized\blue{.} \blue{In this case,} the \blue{approximated} filter will introduce additional MSE $\Delta \beta$:
\begin{align}
\Delta \beta = \beta-\beta_{min},\label{43}
\end{align}
where $\beta$ is the MSE for signals filtered by the approximate MMSE, and $\beta_{min}$ is the minimum MSE filtered by the exact MMSE. \blue{Assume} that the Doppler matrix $\boldsymbol \nu$ is approximated by the largest Doppler index scalar $\nu_{max}$ \blue{in} the worst scenario. We have the following proposition that describes the relationship among $\Delta \beta$, $\Delta m$ and $\nu_{max}$.
\begin{proposition}
The number of approximated signal layer $\Delta m$ is a function of $\nu_{max}$ and $\Delta \beta$, and it is given by
\begin{align}
\Delta m  = \frac{MN(cos^{-1}(\frac{\Delta \beta}{2a})-\theta)}{2\pi\nu_{max}}, \label{44}
\end{align}
where $a$ and $\theta$ are the magnitude and phase of $\mu_{n,m} = {\bf w}^{*}_{n,m}{\bf H}_{n,m}[:,l']$, respectively.
\end{proposition}
\begin{proof}
The proof of (\ref{44}) is provided in Appendix A.
\end{proof}

\begin{figure}[!t]

\includegraphics[width=3.2in]{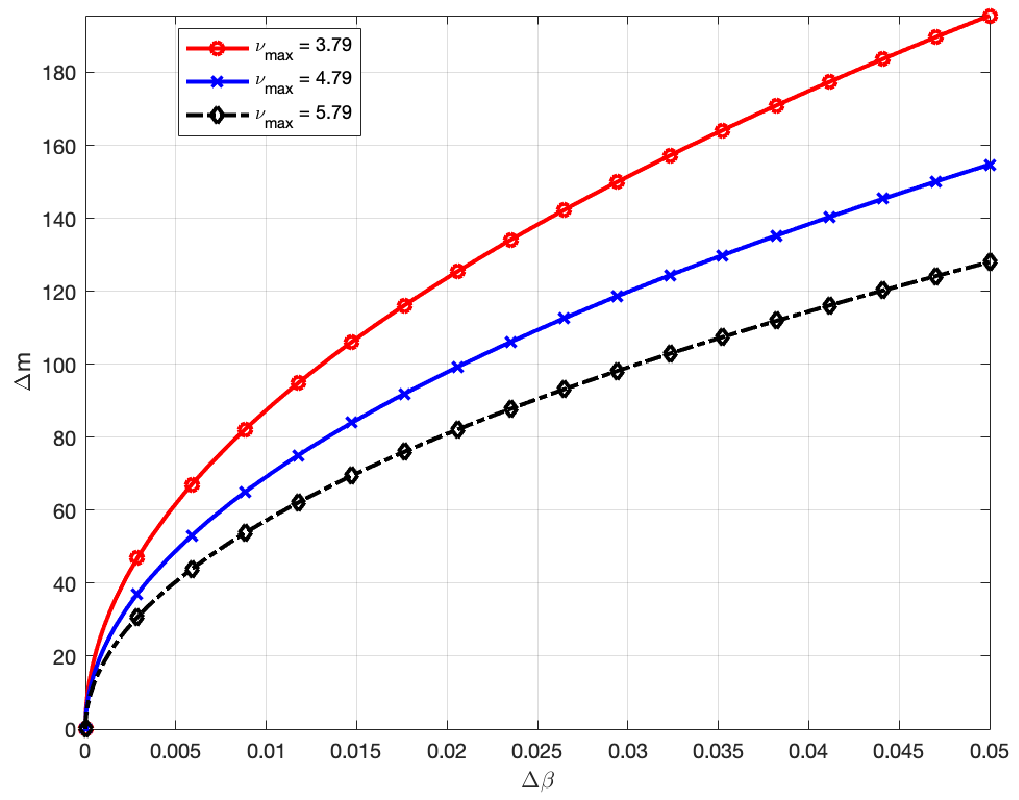}
\caption{Relation between the number of approximated sub-channels $\Delta m$  and the given MSE tolerance $\Delta \beta$, when maximum Doppler index is $\nu_{max}=3.79,  4.79,  5.79$ which correspond to Doppler shift of 888Hz, 1122Hz, and 1357Hz, respectively.}
\label{f:fig1}
\end{figure}

Now we discuss the impact of $\nu_{max}$ and $\Delta \beta$ on the choice of $\Delta m$. \blue{Consider an OTFS system with $M=512$ and $N=128$.} Fig. \ref{f:fig1} and Fig. \ref{f:fig2} show how $\Delta m$ is affected by $\nu_{max}$ or $\Delta \beta$
when one of them is fixed. It can be seen from Fig. \ref{f:fig1} that $\Delta m$ increases with respect to $\Delta \beta$ monotonically when $\nu_{max}$ is fixed. This is quite reasonable as a larger MSE tolerance allows more freedom for more signal layers to be approximated, at the cost of increased errors. On the other hand, when MSE tolerance $\Delta \beta$ is fixed, the number of approximated signal layers $\Delta m$ is inversely proportional to Doppler $\nu_{max}$. It is also a justifiable result since a stronger Doppler effect causes a larger phase shift among sub-channels, which increases the constraints on the number of signal layers that can be approximated by ${\bf w}^{*}_{n,m}$. Define $\Delta m_{c}$ to be the number of symbols within the coherence time - namely ‘coherence symbols’, i.e. $\Delta m_{c}\times(\frac{T}{M}) = T_{c}$, where $\frac{T}{M}$ is the time resolution and $T_{c}$ is the coherence time. The coherence time $T_{c}$ can also be calculated from the inverse of the Doppler spread, i.e. $\frac{1}{B_{d}} = T_{c}$, where $B_{d} = 2\nu_{max}\times\frac{1}{NT}$. Therefore, we arrive at the relation of $\Delta m_{c}$ and $\nu_{max}$
\begin{align}
{\Delta m_{c}}=\frac{MN}{2\nu_{max}}.\label{eq3121}
\end{align}
Substitute (\ref{eq3121}) into (\ref{44}) \blue{gives} the relation between $\Delta m$ and $\Delta m_{c}$
\begin{align}
{\Delta m}=\frac{\Delta m_{c}}{\pi}\left({\cos^{-1}}\left(\frac{\Delta \beta}{2a}-\theta\right)\right).\label{eq3122}
\end{align}
From (\ref{eq3122}), it can be observed that ${\Delta m}$ varies linearly with $\Delta m_{c}$ and they have a positive correlation. This implies that a longer channel's coherence time allows more freedom for the number of approximated sub-channels when the error tolerance $\Delta \beta$ is fixed.  Note that because the size of OTFS frame $M$ and $N$ influence the delay and Doppler resolution for a given system bandwidth and finite duration, they also play a role in determining the number of approximated sub-channels in (\ref{44}) and the coherence symbols in (\ref{eq3121}) and consequently in (\ref{eq3122}). \par
So far we have discussed how to re-use the exact ${\bf w}^{*}_{n,m}$ calculated from the $m$-th sub-channel as an approximated MMSE weights for the rest of $\Delta m$ sub-channels. Under the MSE tolerance $\Delta \beta$ and channel Doppler index $\nu_{max}$ restrictions, the largest number of approximated signal layers allowed $\Delta m$ can then be determined from (\ref{44}). In other words, this $\Delta m$ represents the largest number of subsequent sub-channels that can be approximated by ${\bf w}^{*}_{n,m}$ without exceeding the given error threshold $\Delta \beta$ under the worst assumed channel condition. After this, the detector will re-calculate the exact MMSE filter coefficients ${\bf w}^{*}_{n,\Delta m+m+1}$ at $(\Delta m+m+1)$-th sub-channel and use this newly generated exact MMSE as an approximation for the next $\Delta m$ sub-channels. This process repeats until all the $M-l_{max}$ signals have been detected. The detailed steps of the approximated SIC-MMSE are in Algorithm 1.
{\begin{algorithm}[t]
  \label{BP}

  \caption{Iterative Cross Domain Approximate SIC-MMSE}
  \textbf{Input}: ${\bf r}$, ${\bf H}$, $\tilde{\bf {s}}^{(0)}_{m}$, ${\bf V}$, $E_{s}$, $\Delta m$

  \begin{algorithmic}[1]
    \STATE{{\bf for} $i=1$ to $i_{max}$ }
    \STATE{\quad{\bf for} $m=1$ to $M-l_{max}$}
    \STATE{\quad\quad{\bf for} $n=1$ to $N$}
     \STATE{\quad\quad\quad{\bf if} $m = 1$ or $m = m+\Delta m+1$}
     \STATE{\quad\quad\quad\quad${\bf w}^{(i)}_{n,m}={\bf H}^{H}_{n,m}[:,l']({\bf H}_{n,m}{\bf V}_{n,m}{\bf H}_{n,m}^{H}$\\ \quad\quad\quad\quad\quad\quad\quad+$\sigma_{n}^2{{\bf I}_{l_{\max}+1}})^{-1}$}

     \STATE{\quad\quad\quad{\bf end if}}

     \STATE{\quad\quad\quad${\hat{\bf r}^{(i)}_{n,m}} = \sum_{j=0}^{l'+l_{max}}{\bf H}_{n,m}[:,j]{s}_{n,(m'+j)}+{\bf z}_{n,m}-$\\ \quad\quad\quad\quad\quad\quad$\sum_{j=0}^{l'-1}{\bf H}_{n,m}[:,j]\tilde{s}_{n,(m'+j)}^{(i)}-$}
      \STATE{\quad\quad\quad\quad\quad\quad$\sum_{k=l'+1}^{l'+l_{max}}{\bf H}_{n,m}[:,k]\tilde{s}_{n,(m'+k)}^{(i-1)}$}
      \STATE{\quad\quad\quad${\hat{s}}^{(i)}_{n,m} = {\bf w}^{(i)}_{n,m}{\hat{\bf r}}^{(i)}_{n,m}$}
     \STATE{\quad\quad\quad$\mu^{(i)}_{n,m} = {\bf w}^{(i)}_{n,m}{\bf H}_{n,m}[:,l']$}
     \STATE{\quad\quad\quad$\sigma_{\tilde{z}^{(i)}_{n,m}}^2 = \frac{\mu^{(i)}_{n,m}-|\mu^{(i)}_{n,m}|^2}{|\mu^{(i)}_{n,m}|^2}E_{s}.$}
    \STATE{\quad\quad{\bf end for}}
    \STATE{\quad\quad${\hat{\bf y}}_{m}= {{\bf F}_{N}}{\hat{\bf s}^{(i)}_{m}}$}
    \STATE{\quad\quad${\tilde{{x}}}_{m}[n]= \sum_{a\in \mathcal {Q}}{\text {Pr}}\{\hat{y}_{m}[n]|x_{m}[n]=a\}\times a$}
      \STATE{\quad\quad${{\bf V}}_{m}[n,n]= \sum_{a\in \mathcal {Q}}{\text {Pr}}\{\hat{y}_{m}[n]|x_{m}[n]=a\}$\\ \quad\quad\quad\quad\quad\quad\quad$\times |a-{\tilde{{x}}}_{m}[n]|^2$}
      \STATE{\quad\quad$\tilde{\bf {s}}^{(i)}_{m}= {{\bf F}_{N}^{\dag }}\hat{{\bf x}}_{m}$}\\
      \STATE{\quad\quad${{\bf V}^{(i)}_m}= {{\bf F}_{N}^{\dag}}{\bf V}_{m}{\bf F}_{N}$\\}
    \STATE{\quad{\bf end for}}
    \STATE{{\bf end for}}
  \end{algorithmic}
   {\bf Output:} {\bf {s}} {\bf {x}}.
\end{algorithm}}

\begin{figure}[!t]
\includegraphics[width=3.2in]{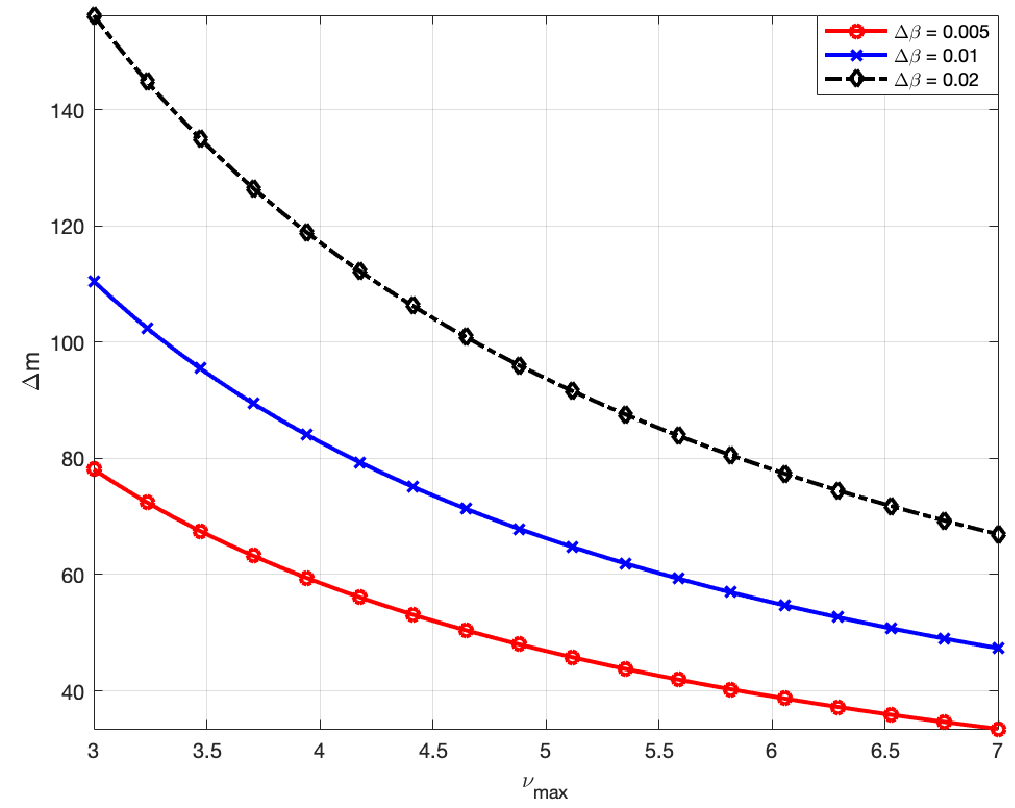}
\caption{Relation between the number of approximated sub-channels $\Delta m$ and the maximum Doppler index $\nu_{max}$, when MSE tolerance $\Delta \beta=0.005, 0.01, 0.02$.}
\label{f:fig2}
\end{figure}

\subsection{SINR Analysis for Approximate SIC-MMSE}
Recall that the output from the exact SIC-MMSE filtering ${\hat{s}}^{(i)}_{n,m}$ can be written in relation to the transmit symbol $s^{(i)}_{n,m}$ in the $i$-th iteration as in (\ref{eq38}) and the variance of ${\tilde {z}^{(i)}_{n,m}}$ can be approximated as in (\ref{eq399}) or it can be computed as
\begin{align}
\sigma^2_{{\tilde{z}_{n,m}^{(i)}}}=&\sum_{j=1}^{l'-1}|{\bf w}^{*}_{n,m}{\bf H}^{H}_{n,m}[:,j]|^2 \sigma^2_{e_{n,(m'+j)}^{(i)}}\nonumber \\&+\sum_{k=l'+1}^{l'+l_{max}} |{\bf w}^{*}_{n,m}{\bf H}_{n,m}[:,k]|^2 \sigma^2_{e_{n,(m'+k)}^{(i-1)}}+||{\bf w}^{*}_{n,m}||^2\sigma_{n}^2,\label{eq46}
\end{align}
where $\sigma^2_{e}$ and $\sigma^2_{n}$ are the error variance and noise variance, respectively. Without MMSE approximation, the SINR from the exact MMSE filtering at the $m$-th sub-channel is calculated as
\begin{align}
\text{SINR}_{n,m} = \frac{|\mu^{(i)}_{n,m}|^2E_s}{\sigma^2_{{\tilde{z}_{n,m}^{(i)}}}},
\end{align}
where $\mu^{(i)}_{n,m}$ is calculated from (\ref{eq399}) with the $m$-th sub-channel's exact MMSE ${\bf w}^{*}_{n,m}$. If ${\bf w}^{*}_{n,m}$ is re-used as an approximation for the $(\Delta m+m)$-th sub-channel, the SINR at the $(\Delta m+m)$-th sub-channel from the approximated MMSE filtering can be calculated as
\begin{align}
\text{SINR}_{n,(\Delta m+m)} = \frac{|\mu^{(i)}_{{n,(\Delta m+m)}}|^2E_s}{\sigma^2_{{\tilde{z}_{n,(\Delta m+m)}^{(i)}}}},
\end{align}
where $\mu^{(i)}_{{n,(\Delta m+m)}}$ and ${\sigma^2_{\tilde{z}_{{n,(\Delta m+m)}}}}$ are filtered by the $m$-th sub-channel MMSE ${\bf w}^{*}_{n,m}$ instead. As such, the non-optimal filter weights ${\bf w}^{*}_{n,m}$ for the $(\Delta m+m)$-th sub-channel causes SINR at the ${(\Delta m+m)}$ sub-channel to drop continuously as $\Delta m$ increases until $\Delta m$ reaches the point where the next exact MMSE needs to be re-calculated. Fig. \ref{f:fig33} shows how SINR changes with respect to $\Delta m$ for the filtered symbols at the $n$-th block under both optimal SIC-MMSE operation and the approximate SIC-MMSE that is proposed in this section. As it can be seen, the SINR for the approximate MMSE drops continuously as the number of approximated $\Delta m$ sub-channel increases and stays below the optimal SINR. It reaches the lowest right before the next exact MMSE will be re-calculated. This also indicates that the MSE at $(m+\Delta m)$-th sub-channel has just reached to/beyond the MSE tolerance threshold $\Delta \beta$. After re-calculating the exact MMSE filter coefficients for the $(m+\Delta m+1)$-th sub-channel, the SINR goes back to the optimal value, which intersects with the exact MMSE SINR curve. The SINR of the approximated SIC-MMSE continues to experience such up-and-down alternations until detection reaches the last $(M-l_{max})$-th symbol in the $n$-th block.

\begin{table*}

\caption{Complexity Comparisons Among Different Detection Algorithms}
\label{my-label}
\begin{adjustbox}{width=1\textwidth}
\begin{tabular}{ | m{4.2em} | m{2.8cm}| m{3.4cm} | m{2.8cm}| m{3.3cm}| m{2.8cm}|}
  \hline
  Detector & Classical MMSE \cite{Li2021CrossModulation}/SIC-MMSE\cite{Chatterjee2021NonorthogonalTransmission} & MP\cite{P.RavitejaKhoaT.PhanQianyuJinYiHong2022} & MRC\cite{Thaj2020LowSystems} & Proposed SIC-MMSE & Proposed Approximate SIC-MMSE\\
  \hline
  Complexity & $\mathcal{O}((M\!-l_{max})^3N^3)$ & $\mathcal{O}((M\!-l_{max})NP\mathcal{|Q|})$ & $\mathcal{O}((M\!-l_{max})NP)$& $\mathcal{O}((M\!-l_{max})Nl_{max}^3)$ & $\mathcal{O}\Bigl(\frac{(M-l_{max})Nl_{max}^3}{\Delta m}\Bigl)$\\
  \hline

\end{tabular}
 \end{adjustbox}
\end{table*}

\begin{figure}[!t]
\includegraphics[width=3.2in]{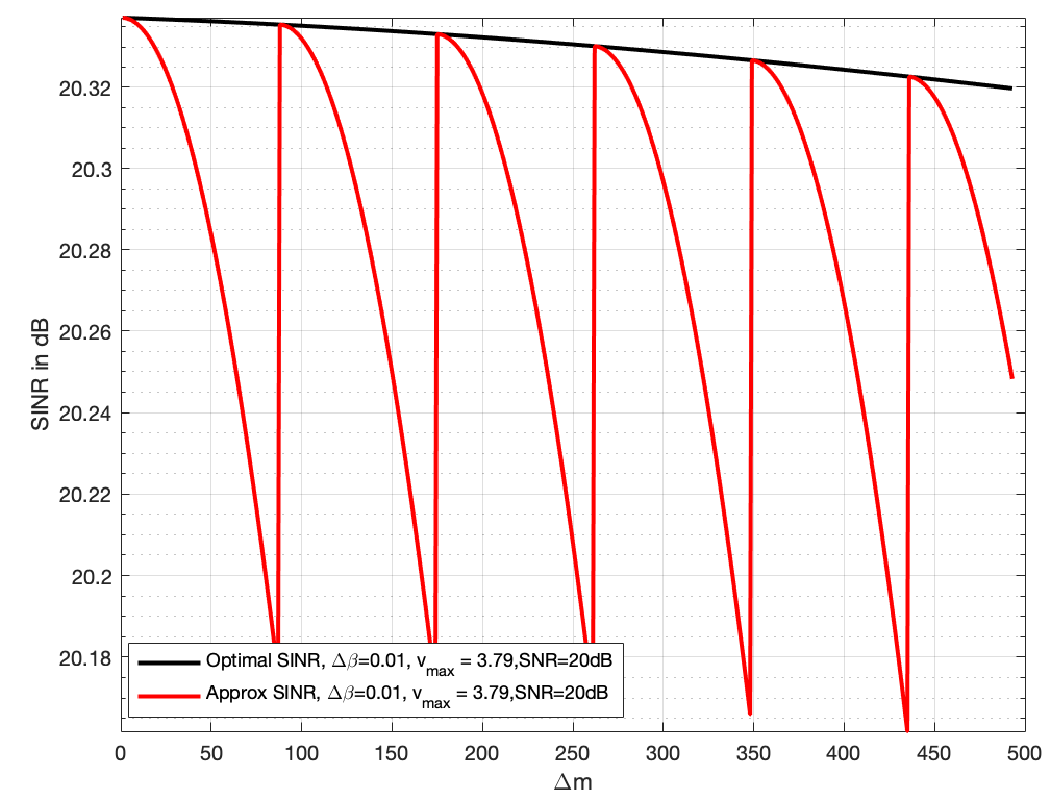}
\caption{SINR for each symbol after SIC-MMSE and the approximate SIC-MMSE filtering for 16QAM OTFS with $M=512, N=128$.}
\label{f:fig33}
\end{figure}

\subsection{Complexity Analysis for Approximate SIC-MMSE Per Iteration}

The computational complexity for interference cancellation in (\ref{eq36}) can be reduced by storing $\Delta {\bf r} = {\bf r}-\sum\limits_{j=1}^{l'+1}{\bf H}_{n,m}[:,j]{\tilde s}_{j}^{(i)}$ as an interference-reduced variable so that ${\bf r}$ does not have to cancel interference repeatedly each time for new target symbols. As a result, this only requires $\mathcal{O}((M-l_{max})Nl_{max})$ complex multiplications. Without using the MMSE approximation, the overall complexity is dominated by the MMSE weights calculation, which requires an overall complexity of $\mathcal{O}((M-l_{max})Nl_{max}^3)$. When approximation is applied, this can be significantly reduced to $\mathcal{O}\Bigl(\frac{(M-l_{max})Nl_{max}^3}{\Delta m}\Bigl)$. By (\ref{eq3122}), the overall complexity is then $\mathcal{O}\Bigl(\frac{(M-l_{max})Nl_{max}^3\pi}{(\Delta m_{c}{\cos^{-1}}(\frac{\Delta \beta}{2a}-\theta))}\Bigl)$. It can be seen that the complexity is also determined by the channel coherence time - a longer coherence time allows more freedom for the number of approximated sub-channels. This can result in only a few sub-channels being required to calculate the exact MMSE weights. In our simulation in Sec. VI, we will show that $(M-l_{max})=493$ sub-channels at one particular block can be approximated by only calculating less than $5$ exact MMSE weights without any noticeable performance degradation. As for the channel parameters, the maximum Doppler index in the channel is $\nu_{max}=3.79$ which corresponds to a Doppler spread $B_{d}=888Hz$ approximately. Under such a time-selective channel, the coherence time calculated from the Doppler spread is around $T_{c}=0.0011s$, and the number of coherence symbols $\Delta m_{c}$ is around $\Delta m_{c}=8645$, which renders the approximated sub-channels to sit well within the coherence time. Apart from MMSE weight calculations, computing the MMSE output ${\hat{s}}^{(i)}_{n,m}$ and variance $\sigma_{\tilde{z}^{(i)}_{n,m}}^2$ for all symbols requires another $\mathcal{O}(2(M-l_{max})Nl_{max})$ complex multiplications. The transformation between the time domain and the DD domain needs two $N$-FFT operations which have a complexity of $\mathcal{O}(2(M-l_{max})Nlog_{2}N)$. Even so, the overall computation complexity is still dominated by $\mathcal{O}\Bigl(\frac{(M-l_{max})Nl_{max}^3}{\Delta m}\Bigl)$. Especially when $\Delta m$ is large, it is comparable to the complexity of MRC that is dominated by an overall complexity of $\mathcal{O}((M-l_{max})NP)$ \cite{LR30}. The complexity comparison between the proposed detectors and some benchmark detectors for ZP-OTFS are summarized in Table 1. \par
{Here, we provide a numerical example of the complexity for the proposed algorithms compared with the other algorithms. For instance, the frame size is set as $M=512, N=128$ and the largest delay spread is $l_{max} = 19$. For classical MMSE and SIC-MMSE, the complexity order is thus $2.8\times 10^{14}$. However, for our proposed soft SIC-MMSE, this complexity order can be reduced to $4.5\times 10^8$, and for the approximate SIC-MMSE, this complexity order can be further brought down to $4.5\times 10^6$. The complexity of MRC is $5.8 \times 10^{5}$, which is now comparable to our proposed approximate SIC-MMSE. However, the proposed SIC-MMSE has almost 2.5 dB gain over MRC for 16QAM modulation at BER $3.5\times 10^{-4}$ as shown in our simulation.}
\subsection{Performance Analysis}
In this subsection, we analyze the error performance of the proposed approximate iterative SIC-MMSE using state evolution (SE). We demonstrate that SE can reliably predict/analyze the mean squared errors during the iterative process. We will derive upper and lower bounds on SE based on the residual interference power. Simulation results show that the MSE of our approximate SIC-MMSE closely matches the theoretical MSE lower bound obtained from SE.\par
We define the error state at the $i$-th iteration after SIC-MMSE linear detection in the time domain as
\begin{align}
\tau^2[i] \delequal \lim_{MN \to \infty} \frac{1}{MN}{\text {Tr}}({\bf C}^{t}), \label{eq4}
\end{align}
where ${\bf C}^{t}$ is the error covariance matrix for all $MN$ time domain symbols after SIC-MMSE detection. Each element in ${\bf C}^{t}$ at the $n$-th block can be calculated using (\ref{eq399}) such that
\begin{align}
C^t[{m},{m}] = \frac{\mu_{{m}}-|\mu_{{m}}|^2}{|\mu_{{m}}|^2}E_{s}, \label{eq7}
\end{align}
We set a lower threshold of $\tau_{low}^2[i]$ by assuming the residual interference in the current iteration has been removed completely with ${\bf V}_{n,m}={\text {diag}}\{0,...,0,\tau_{prior}^2[i],...,\tau_{prior}^2[i]\}$. We also set an upper threshold $\tau_{up}^2[i]$ by assuming the residual interference was not canceled which still have full signal power $E_s$, i.e. ${\bf V}_{n,m}={\text {diag}}\{E_s,...,E_s,\tau_{prior}^2[i+1],...,\tau_{prior}^2[i+1]\}$, where $\tau_{prior}^2[i+1]$ is the time domain \emph{a prior} error state that was obtained based on DD domain detection from the last iteration, which will be introduced in more details later.\par
Due to the unitary transformation between time domain and DD domain via FFT/IFFT, the covariance of the equivalent noise in the DD domain is assumed to be the same as the time domain error covariance. The DD domain observations can be expressed as
\begin{align}
{\bf y} = {\bf x}+{\bf n}, \label{eq8}
\end{align}
where ${\bf y}$ is the DD domain observations transformed from time domain estimates after the SIC-MMSE, i.e., ${\bf y} = {\bf F}_N{\bf \hat {s}}$, and ${\bf x}$ is the transmit symbols in DD domain with a given constellation set $\mathcal {Q}$ and {\bf n} is the equivalent Gaussian noise in the DD domain with the same covariance matrix as ${\bf C}^{t}$. Here we also define the effective SNR in DD domain as
\begin{align}
\text{SNR}_{\text{eff}} = \frac{E_s}{\tau^2[i]}, \label{eq81}
\end{align}
After we perform the non-linear detection in the DD domain, the error covariance will be updated. Based on this, we define the second error state at the $l$-th iteration after non-linear symbol-by-symbol MAP detection in the DD domain as
\begin{align}
\nu^2[i] \delequal \lim_{MN \to \infty} \frac{1}{MN}{\text Tr}({\bf C}^{DD})  = \text{MSE}({\tau^2[i]}),  \label{eq9}
\end{align}
which can be interpreted as the average of the $MN$ \emph{a posteriori} variances. Note that the covariance matrix ${\bf C}^{DD}$ cannot be
directly derived in terms of $\tau^2[i]$ due to the DD domain non-linear detection, we instead consider a Monte Carlo approach
to calculate the MSE. A sufficiently large value of $MN$ observations are manually generated according to (\ref{eq8}), given the variance of ${\bf n}$ as $\tau^2[i]$. Then, the MSE can be obtained by
\begin{align}
\text{MSE} = {\mathbb {E}}\{|x-{\mathbb {E}}\{{x|x+n}\}|^2\}.\label{eq10}
\end{align}
By averaging the errors across the $MN$ samples, we can then obtain the error state $\nu^2[i]$ after DD domain non-linear estimation. Finally, the error state $\nu^2[i]$ will also be used as the initial time domain error state at $(i+1)$-th iteration for calculating the ${\bf C}^{t}$ in (\ref{eq7}),
\begin{align}
\tau_{prior}^2[i+1] = \nu^2[i]. \label{eq11}
\end{align}\par
Fig. \ref{f:fig4} shows the simulated MSE by SE analysis. For the simulation, MSE statistics were collected by averaging the errors prior to time domain linear estimation. {As depicted in the plot, the MSE exhibits a consistent decrease as the number of iterations progresses, typically converging within just four iterations. This trend indicates that the residual interference power diminishes significantly with increasing iterations, thereby leading to an enhancement in overall error performance. Moreover, the analysis underscores a close alignment between the simulation results and the lower bound analysis, assuming complete interference cancellation.} The close match between the simulation and lower threshold SE also suggests that the residual interference power is negligible. Fig. \ref{f:fig5} displays \blue{$\text{SNR}_{\text{eff}}$} prior to the DD domain symbol-by-symbol estimation. The effective SNR increases in the first few iterations and then stabilizes as a result of a more suppressed residual interference plus noise power with more iterations.

\begin{figure}[!t]
\includegraphics[width=3.5in]{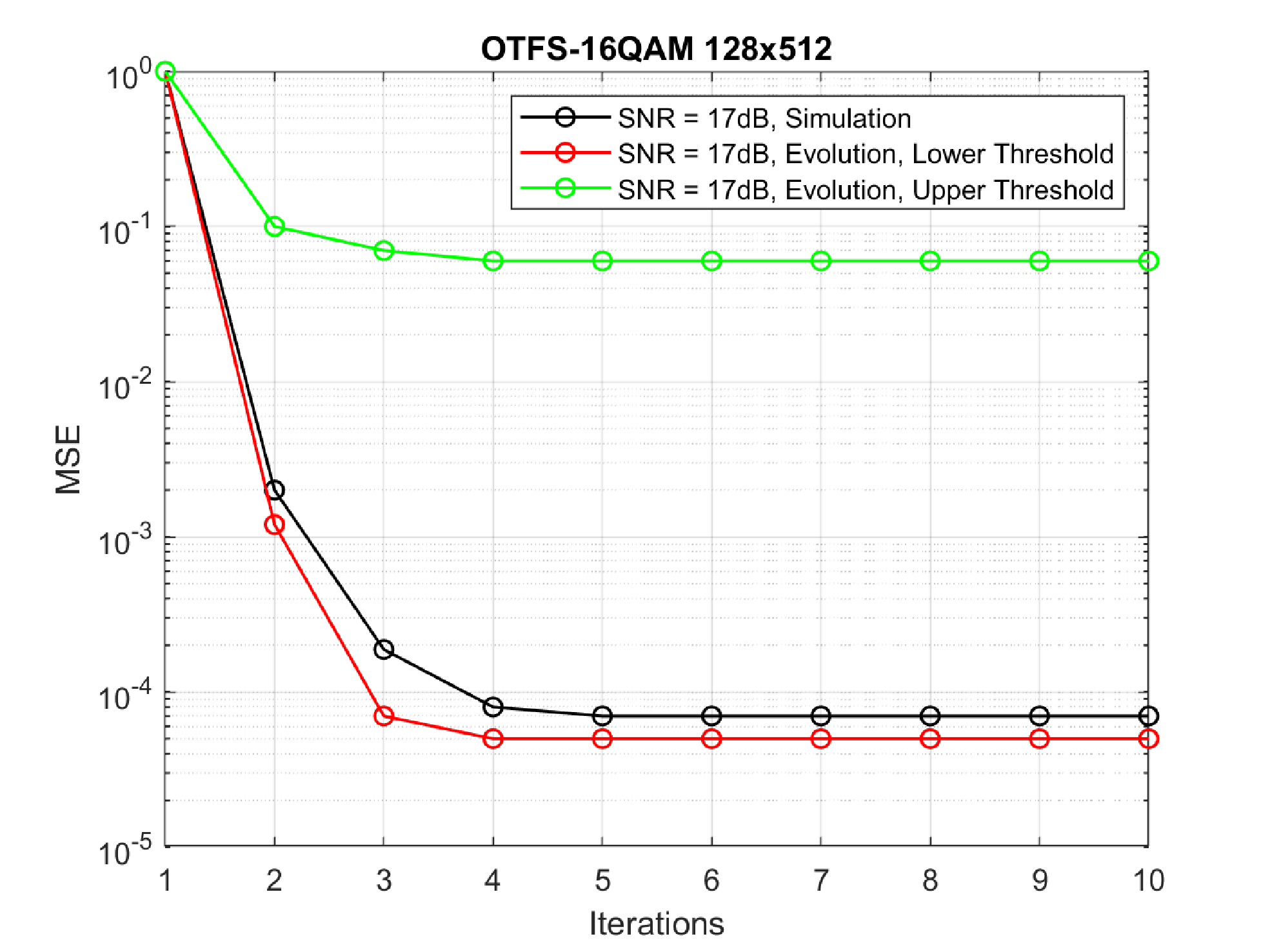}
\caption{MSE performance for both simulations and SE analysis across different iterations for 16QAM OTFS with $M=512, N=128$, SNR$=17$dB and EVA channel.}
\label{f:fig4}
\end{figure}

\begin{figure}[!t]
\includegraphics[width=3.5in]{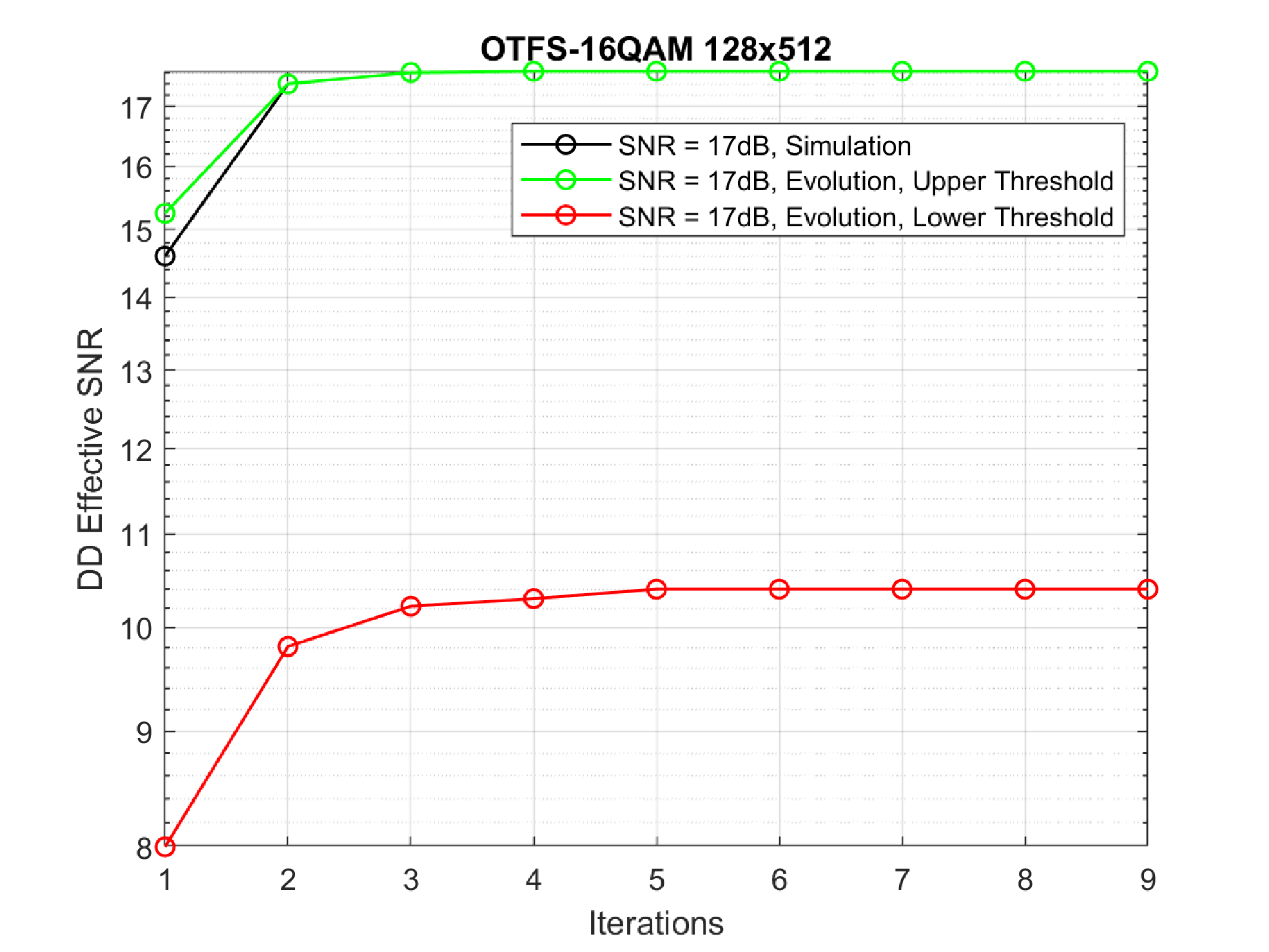}
\caption{Effective SNR in the DD domain for both simulations and SE analysis across different iterations for 16QAM OTFS with $M=512, N=128$, SNR$=17$dB and EVA channel.}
\label{f:fig5}
\end{figure}

\section{SISO Iterative Decoding and Detection}
This section investigates the proposed low-complexity SIC-MMSE based turbo receiver, in which the soft extrinsic information is exchanged iteratively between a decoder and an approximate SIC-MMSE detector.\par

\subsection{Iterative receiver}
Fig. \ref{f:fig3} shows the block diagram for the proposed turbo receiver. To execute the turbo process, the detector receives the time domain input symbols from the channel and performs the low-complexity approximate SIC-MMSE proposed in Section IV. After all the time domain SIC-MMSE outputs have been processed, the \emph{a posteriori} mean and variance for each transmit \blue{symbol of $\mathcal{Q}$} can be calculated by (\ref{eq312}) and (\ref{eq313}). The log-likelihood ratio (LLR) of the $p$-th bit $b_p$ in symbol ${\tilde{{x}}}_{m}[n]$ detected by SIC-MMSE in ${\tilde{\bf x}}_{m}$ can be calculated as
\begin{align}
L_{m,n}[p] =& \log\frac{\text {Pr}(b_{p}=0|\hat{y}_{m}[n])}{\text {Pr}(b_{p}=1|\hat{y}_{m}[n])} \nonumber \\=& \log\frac{\sum_{q\in \mathcal{Q}^{p}_0}\text {exp}(-|\hat{y}_{m}[n]-q|^2/\tilde{{\bf V}}_{m}[n,n])}{\sum_{q'\in \mathcal{Q}^{p}_1}\text {exp}(-|\hat{y}_{m}[n]-q'|^2/\tilde{{\bf V}}_{m}[n,n])},\label{eq51}
\end{align}
where $\mathcal{Q}^{p}_0$ \blue{and $\mathcal{Q}^{p}_1$ are the symbol sets} that has the target $p$-th bit \blue{beginning with 0 and 1, respectively.} The complexity of LLR calculations can be reduced by using approximations in \cite{Studer2011ASICCancellation}. After acquiring the intrinsic LLR for all the bits, we then have the LLR vector ${\bf L}^{\text{out}}_{c}$ \blue{from the detector in serial}. The extrinsic LLR is then computed by
\begin{align}
{\bf L}_{c}^e = {\bf L}^{\text{out}}_{c}-{\bf L}^{a},\label{eq52}
\end{align}
where ${\bf L}^{a}$ is the \emph{a priori} LLR to the detector. \blue{The} extrinsic LLR \blue{is} de-interleaved resulting in \blue{$\pi^{-1}({\bf L}_{c}^e)= {\bf L}_{c}^{a}$, where $\pi(.)$ and $\pi^{-1}(.)$ represent a pair of interleaving and de-interleaving functions}. These LLR will be used as the input LLR to the decoder. The output LLR ${\bf L}^{\text{out}}_{d}$ from the decoder will be re-interleaved becoming \blue{$\pi({\bf L}^{\text{out}}_{d})={\bf L}^{a}$ and passed to the detector.}
\begin{figure}[!t]
\includegraphics[width=\linewidth]{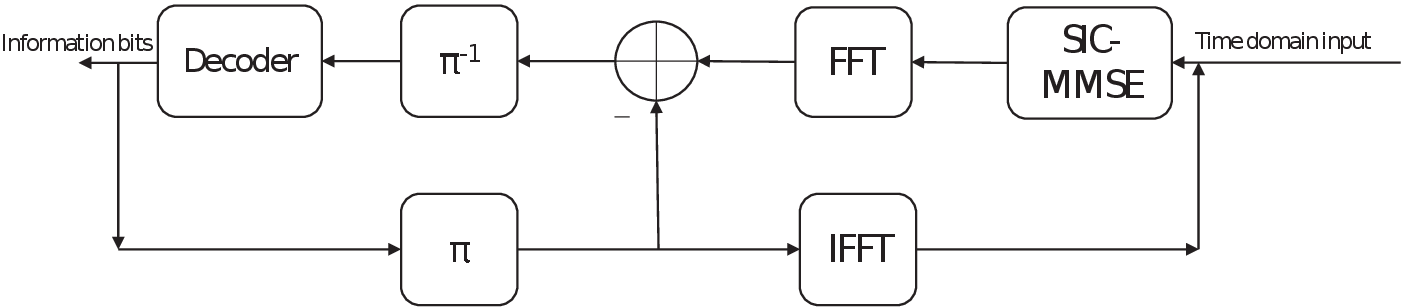}
\caption{Block diagram for the proposed turbo joint detection and decoding.}
\label{f:fig3}
\end{figure}

\blue{As pointed out in \cite{Witzke2002IterativeDetectors}}, \blue{using} the extrinsic information at the output of the decoder performs much worse than intrinsic information, especially for high interference channels \cite{Shen2022ErrorReceivers}. Therefore, we use the intrinsic information from the decoder instead - feeding ${\bf L}^{a}$ directly back into the interleaver as the \emph{a priori} information to the detector. This is possible as we have iterative operations with FFT/IFFT \cite{Li2021CrossModulation}. After all the bits have been interleaved, the output LLR {${{\bf L}}^{\text{a}}$} are converted to soft symbol estimates by
\begin{align}
  {\hat{{x}}}_{m}[n] = \sum_{a\in \mathcal{Q}}a\times \frac{1}{2^{\text{log}_{2}|\mathcal{Q}|}} \prod_{p=1}^{\text{log}_{2}|\mathcal{Q}|}\Bigl(1+\text{sgn}(b_{p}){\text {tanh}}\frac{{L}_{m,n}^{\text{a}}[p]}{2}\Bigl),\label{eq54}
\end{align}
where $\text{sgn}(b_{p})=1$ for $b_p=1$ and $\text{sgn}(b_{p})=-1$ for $b_p=0$. The variance can then be computed by
\begin{align}
\sigma^2_{x_{m}[n]} =& \sum_{a\in \mathcal{Q}}|a|^2\times \frac{1}{2^{\text{log}_{2}|\mathcal{Q}|}} \prod_{p=1}^{\text{log}_{2}|\mathcal{Q}|}\Bigl(1+\text{sgn}(b_{p}){\text {tanh}}\frac{{L}^{\text{a}}_{m,n}[p]}{2}\Bigl)\nonumber\\&-|{\hat{{x}}}_{m}[n]|^2.\label{eq55}
\end{align}
Then, the DD domain soft symbols ${\hat{\bf x}}_{m}$ and their variances will be transformed back into the time domain. They will be used as the \emph{a priori} mean and variance for detection in the next turbo iteration. Intuitively, the newly acquired soft information for time domain symbols generated from the decoder provides more accurate \emph{a priori} information for detection in the following iteration.

\section{Numerical Results}
In this section, we provide simulation results with different system parameters for the proposed detection algorithms. The sub-carrier spacing $\Delta f$ is set as 15 kHz and carrier frequency $4$ GHz. The OTFS frame size has $M=512$ and $N=128$. The wireless channel is generated according to the standard EVA channel model with a UE speed at 120 km/h. The maximum integer delay tap is $l_{max}=19$ and the maximum Doppler tap is $k_{max}=3.79$. The Doppler taps for the $i$-th path are generated from a uniform distribution $U(-k_{max},k_{max})$. Both perfect and imperfect channel state information (CSI) is assumed at the receiver.

\subsection{Perfect Receiver CSI}

\begin{figure}[!t]
\includegraphics[width=3.5in]{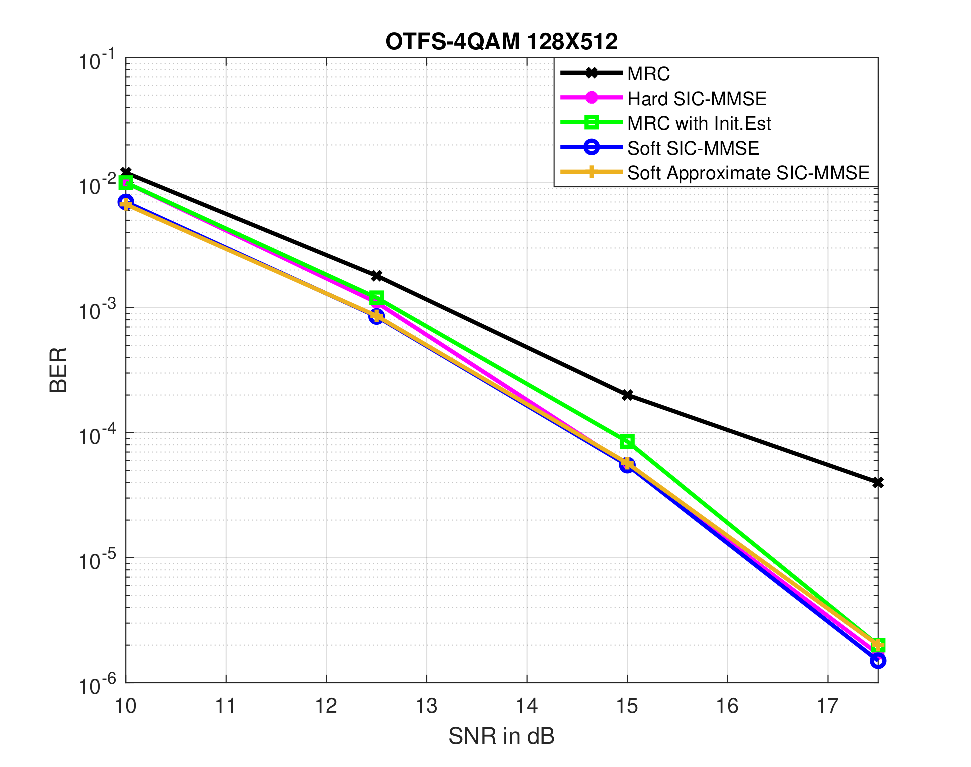}
\caption{BER performance for the proposed SIC-MMSE and approximate SIC-MMSE compared with other detectors for 4QAM OTFS with $M=512, N=128$, EVA channel.}
\label{f:fig6}
\end{figure}

Fig. \ref{f:fig6} \blue{shows the BER} of the proposed SIC-MMSE detection algorithm compared to \blue{several benchmark} OTFS detectors. For iterative detectors, a maximum of 10 iterations is allowed. The detectors included in the comparison are MRC, the hard SIC-MMSE from Section III, the soft SIC-MMSE from Section III, and the soft approximate SIC-MMSE from Section IV. The exact SIC-MMSE detection, which processes all sub-channels with the exact MMSE filtering at high complexity, outperforms MRC without initial estimates by up to 2dB gain, even though our detection does not require time-frequency domain initial estimates. The low-complexity approximate SIC-MMSE does not degrade the error performance compared to the exact SIC-MMSE; however, the complexity is significantly reduced as analyzed in Section IV. The parameters we use for the approximate SIC-MMSE is $\Delta \beta = 0.01$ corresponding to $\Delta m=100$, in which the approximated sub-channels sit well within the coherence time of the physical channel. Overall, compared to other detectors in Fig. \ref{f:fig6}, our approximate SIC-MMSE not only shows better error performance but also has low computational complexity.\par

\begin{figure}[!t]
\includegraphics[width=3.5in]{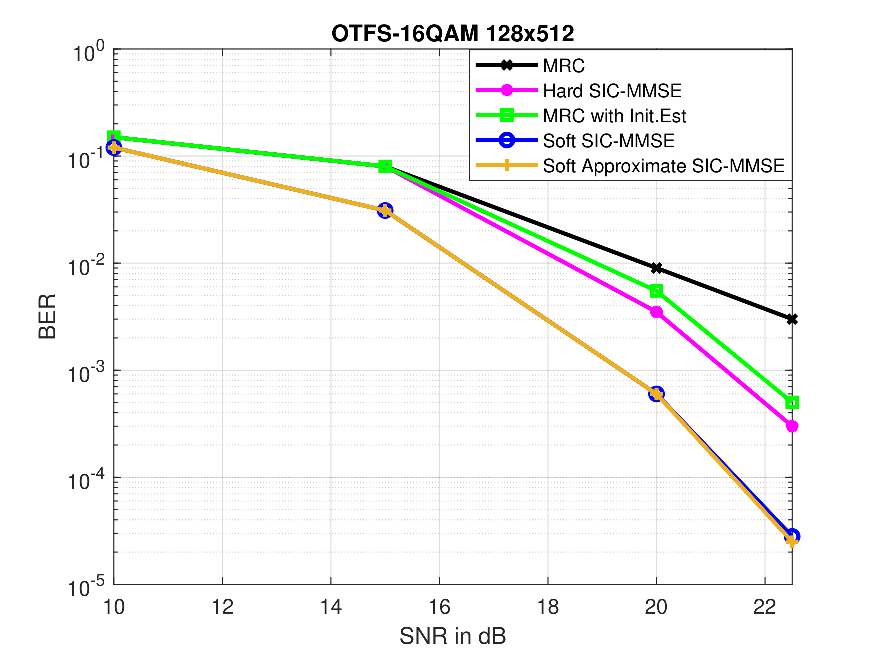}
\caption{BER performance for the proposed SIC-MMSE and approximate SIC-MMSE compared with other detectors for 16QAM OTFS with $M=512,N=128$, EVA channel.}
\label{f:fig7}
\end{figure}

Fig. \ref{f:fig7} displays the BER for 16QAM. Similar to the results obtained for 4QAM, both our proposed SIC-MMSE and approximate SIC-MMSE outperform MRC to an even greater extent. The approximate SIC-MMSE still maintains the same error performance as the non-approximate SIC-MMSE introduced in Section III.

\begin{figure}[!t]
\includegraphics[width=3.7in]{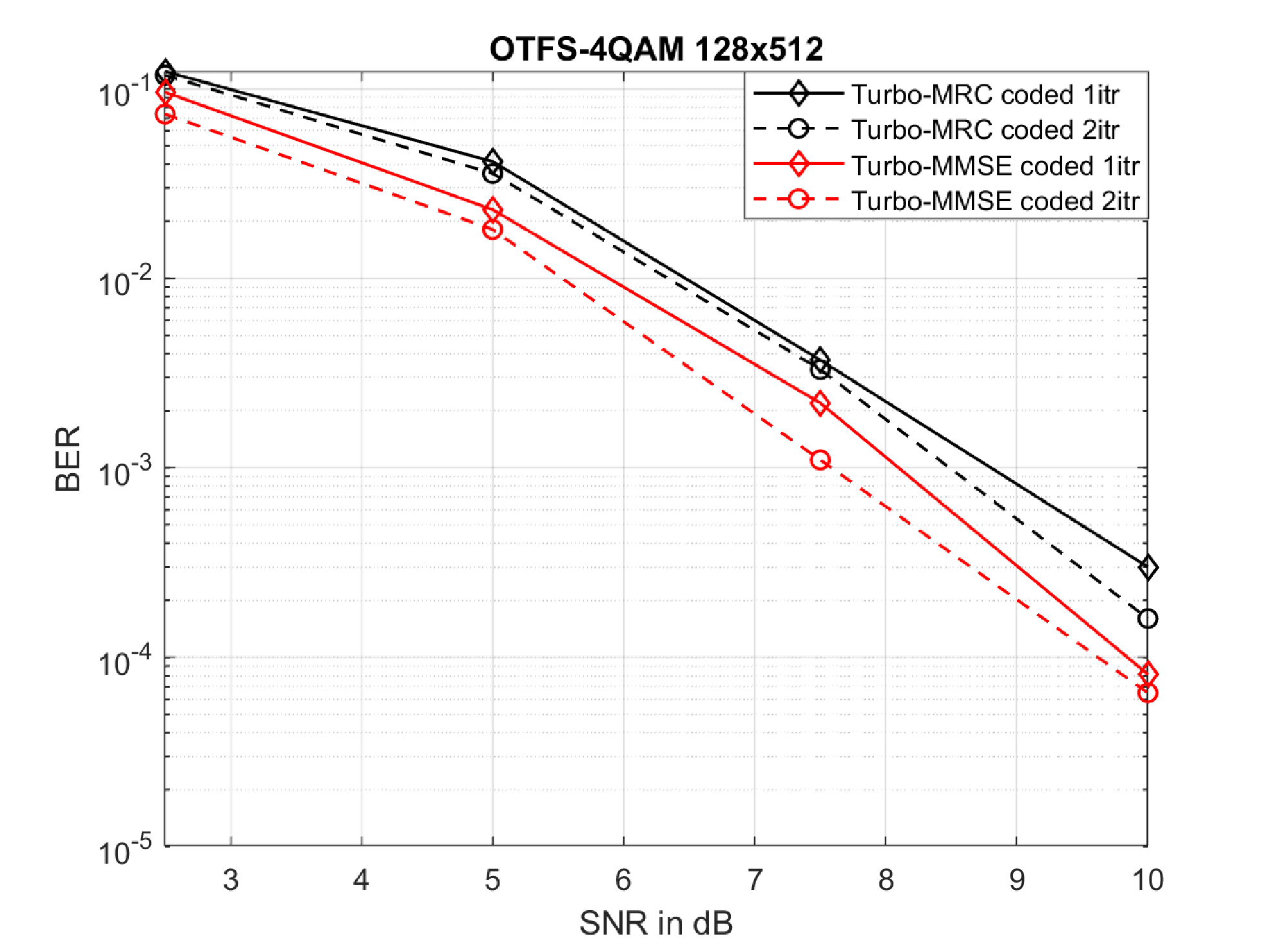}
\caption{BER performance for the proposed turbo joint detection and decoding receiver for 4QAM OTFS with $M=512,N=128$, EVA channel.}
\label{f:fig8}
\end{figure}

Fig. \ref{f:fig8} shows the BER for our proposed turbo iterative detection and decoding algorithm. For the coding scheme, we use a low-density parity-check (LDPC) code. We employ a half-rate LDPC code of length $N_{c}=3840$ bits from \cite{Nguyen2019EfficientRadio}, and every OTFS frame contains {$\lfloor NMlog_{2}(|\mathcal{Q}|)/N_{c}\rfloor$} codewords. The detector performs the approximate SIC-MMSE with the same parameters used in the previous simulations and then exchanges the extrinsic information with the LDPC decoder as introduced in Section V. The BER plot illustrates that one iteration of our proposed turbo approximate SIC-MMSE can outperform two-iteration turbo MRC up to a 1dB gain.

\begin{figure}[!t]
\includegraphics[width=3.7in]{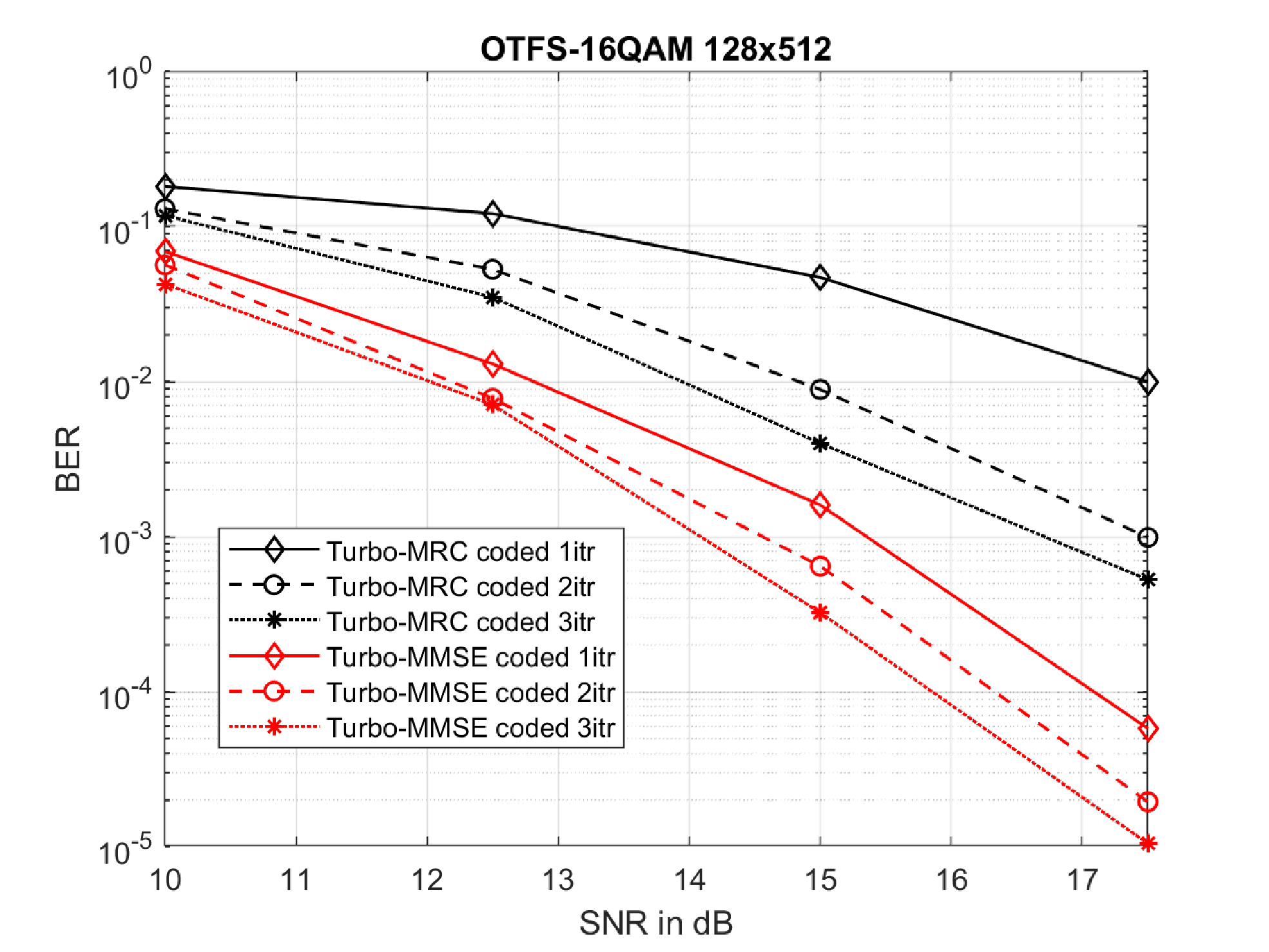}
\caption{BER performance for the proposed turbo joint detection and decoding receiver for 16QAM OTFS with $M=512,N=128$, EVA channel.}
\label{f:fig10}
\end{figure}

Fig. \ref{f:fig10} shows the BER for the 16QAM OTFS. We can see that compared to 4QAM, the detection plus decoding gain of the proposed turbo approximate SIC-MMSE has risen significantly. One-iteration turbo approximate SIC-MMSE can outperform three-iteration turbo MRC with up to 2dB in the BER plot. When three iterations are equally generated for turbo approximate SIC-MMSE, up to 3dB gain can be achieved compared to MRC. In conclusion, our proposed joint iterative detection and decoding algorithm achieved significant performance gain compared to its MRC counterpart, especially for high-order modulation.
\subsection{Imperfect Receiver CSI}
{For imperfect CSI, we consider that the estimated Doppler taps $\boldsymbol{\tilde\kappa}$ have a Gaussian distributed error, i.e. $\boldsymbol{\tilde\kappa} = \boldsymbol{\kappa} + {\Delta{\bf \boldsymbol{\kappa}}}$, where $\boldsymbol{\tilde\kappa}$ is the estimated Doppler taps, and ${\Delta{\bf \boldsymbol{\kappa}}}$ is the estimation error, following the i.i.d Gaussian distribution $\mathcal{CN}(0,\sigma_{k}^2)$. Also, the estimation of the channel coefficients in Fig. \ref{f:fig22} are considered as, i.e. the channel estimation only estimates the first column ${\bf H}_{n,m}[:,1]$ in the time domain which are assumed to have Gaussian distributed errors, i.e. ${\bf \tilde{H}}_{n,m}[:,1] = {\bf {H}}_{n,m}[:,1]+\Delta{{\bf {H}}_{n,m}[:,1]}$ where $\Delta{{\bf {H}}_{n,m}[:,1]}$ is the estimated channel errors following the i.i.d Gaussian distribution $\mathcal{CN}(0,\sigma_{h}^2)$, and the rest of the columns can be generated from (\ref{eq235}), according to the estimated Doppler taps $\boldsymbol{\tilde\kappa}$. }
\begin{figure}[!t]
\includegraphics[width=3.7in]{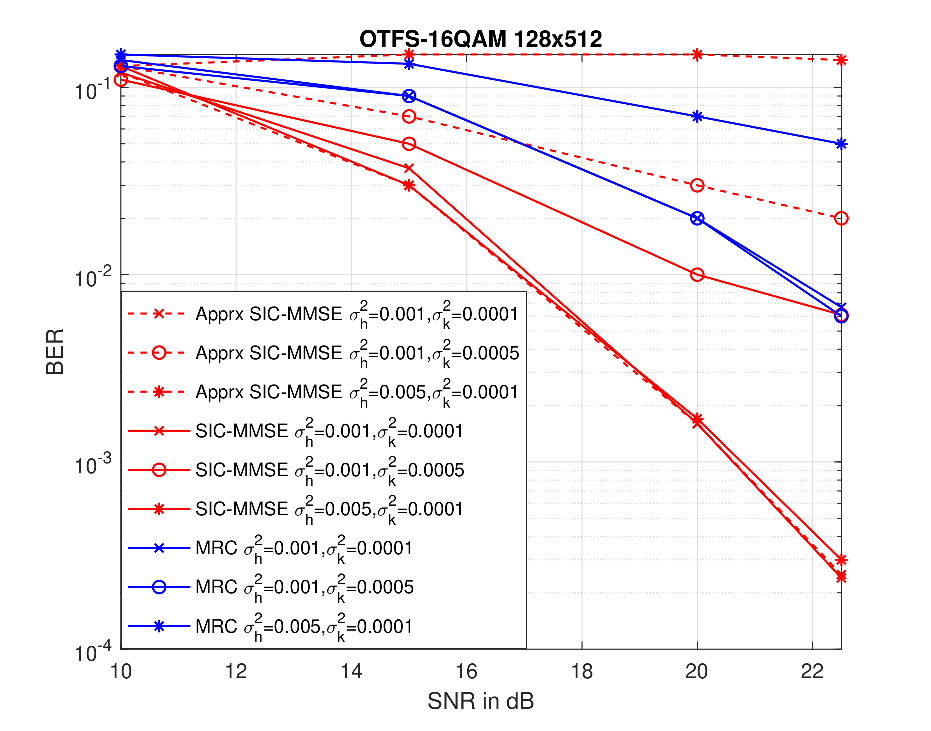}
\caption{BER performance for the proposed SIC-MMSE and approximate SIC-MMSE under different imperfect EVA channels.}
\label{f:fig72}
\end{figure}

{Under imperfect CSI, the quality of the channel estimation has a direct impact on the system's performance. Specifically, as shown in Fig. \ref{f:fig72}, we note that both $\sigma_{h}^2$ and $\sigma_{k}^2$ of the estimation error can impact our proposed approximate MMSE and exact MMSE's BER performance. When $\sigma_{h}^2$ and $\sigma_{k}^2$ are small, i.e. $\sigma_{h}^2=0.001,\sigma_{k}^2=0.0001$, the approximate and exact MMSE share nearly identical BER performance. However, when either $\sigma_{h}^2$ or $\sigma_{k}^2$ is increased, both the approximate and exact MMSE showcased worse BER performance and the approximate MMSE is impacted by the imperfect channel condition more severely. In addition, it is noteworthy that our proposed exact MMSE consistently demonstrates superior performance compared to MRC under the same imperfect channel condition, whereas the approximate MMSE can outperform MRC under some imperfect channel conditions. } \par

\section{Conclusion}
In this paper, we proposed an approximate SIC-MMSE detection algorithm for ZP-OTFS modulation system. Firstly, we derived the hard decision SIC-MMSE and demonstrated that its performance can be improved by soft estimates. Building on the soft estimates, we then proposed an approximate SIC-MMSE algorithm that can significantly reduce the computational complexity while maintaining the same error performance. The maximum number of approximated sub-channels $\Delta m$, which determines the overall computation complexity, can be determined by the pre-set MSE threshold $\Delta \beta$. We analyzed the SINR of each output signal from the approximated MMSE filter. Furthermore, we incorporated the turbo principle into our proposed detector with an LDPC decoder. Note that the OTFS signals considered in this paper do not have received filters, while in practical systems, received filters need to be considered \cite{Shen2022ErrorReceivers}. We would like to point out that the proposed detection method in this paper can also be applied to ODDM signals. Future works will consider these topics.

\appendix
\subsection{Proof of Proposition 1}
The MMSE output ${\tilde{s}_{n,m+\Delta m}}$ filtered by the approximated filtering coefficients ${\bf w}^{*}_{n,m}$ can be written as a linear function with respect to input $s_{n,m+\Delta m}$
\begin{align}
\tilde{s}_{n,m+\Delta m} =& {\bf w}^{*}_{n,m}{\bf H}_{n,m+\Delta m}[:,l']s_{n,m+\Delta m}\nonumber\\&+\sum_{j=1}^{l'-1} {\bf w}^{*}_{n,m}{\bf H}^{H}_{n,m+\Delta m}[:,j]e_{n,(m'+j)}^{(i)}\nonumber\\&+\sum_{k=l'+1}^{l'+l_{max}} {\bf w}^{*}_{n,m}{\bf H}_{n,m+\Delta m}[:,k]e_{n,(m'+k)}^{(i-1)}\nonumber\\&+{\bf w}^{*}_{n,m}{\bf z}^{(i)}_{n,m+\Delta m},\label{eqapp1}
\end{align}
where ${\bf w}^{*}_{n,m}$ is the exact MMSE filter weights calculated from the $m$-th sub-channel and used as an approximation for the $({\Delta m}+m)$-th sub-channel, $m'=\max\{m+\Delta m-l_{max}-1,0\}$ and $l'=\min\{m+\Delta m,l_{max}\}+1$, $\bf z$ is noise. The mean squared error for the non-optimal output $\tilde{s}_{n,m+\Delta m}$ can then be calculated as
\begin{align}
\beta =& \mathbb {E}(|\tilde{s}_{n,m+\Delta m}-s_{n,m+\Delta m}|^2).\label{eqapp2}
\end{align}
Substituting (\ref{eqapp1}) into (\ref{eqapp2}), we then have
\begin{align}
\beta = &\mathbb {E}(|{\bf w}^{*}_{n,m}{\bf H}_{n,m+\Delta m}[:,l']s_{n,m+\Delta m}\nonumber\\&+\sum_{j=1}^{l'-1} {\bf w}^{*}_{n,m}{\bf H}^{H}_{n,m+\Delta m}[:,j]e_{n,(m'+j)}^{(i)}\nonumber\\&+\sum_{k=l'+1}^{l'+l_{max}} {\bf w}^{*}_{n,m}{\bf H}_{n,m+\Delta m}[:,k]e_{n,(m'+k)}^{(i-1)}\nonumber\\&+{\bf w}^{*}_{n,m}{\bf z}^{(i)}_{n,m+\Delta m}-s_{n,m+\Delta m}|^2).\label{eqapp3}
\end{align}
Assuming that time domain random variable $s_{n,m}$ is Gaussian distributed and independent with $e_{n,m}$, (\ref{eqapp3}) can be simplified to
\begin{align}
\beta = &\mathbb {E}(|{\bf w}^{*}_{n,m}{\bf H}_{n,m+\Delta m}[:,l']|^2|{s_{n,m+\Delta m}|^2}\nonumber\\
&-{\bf w}^{*}_{n,m}{\bf H}_{n,m+\Delta m}[:,l']|{s_{n,m+\Delta m}|^2} \nonumber \\
&-\bar{{\bf w}}^{*}_{n,m}\bar{{\bf H}}_{n,m+\Delta m}[:,l']|{s_{n,m+\Delta m}|^2}\nonumber\\
&+\sum_{j=1}^{l'-1} |{\bf w}^{*}_{n,m}{\bf H}^{H}_{n,m+\Delta m}[:,j]|^2|e_{n,(m'+j)}^{(i)}|^2\nonumber \\
&+\sum_{k=l'+1}^{l'+l_{max}}|{\bf w}^{*}_{n,m}{\bf H}_{n,m+\Delta m}[:,k]|^2|e_{n,(m'+k)}^{(i-1)}|^2\nonumber\\
&+||{\bf w}^{*}_{n,m}||^2{\bf z}^{(i)}_{n,m+\Delta m}),\label{eqapp44}
\end{align}
where $\bar{{\bf w}}^{*}_{n,m}\bar{{\bf H}}_{n,m+\Delta m}[:,l']$ is the conjugate of ${\bf w}^{*}_{n,m}{\bf H}_{n,m+\Delta m}[:,l']$. (\ref{eqapp44}) can then be further evaluated
\begin{align}
\beta =& -{\bf w}^{*}_{n,m}{\bf H}_{n,m+\Delta m}[:,l']{E_s}-\bar{{\bf w}}^{*}_{n,m}\bar{{\bf H}}_{n,m+\Delta m}[:,l']{E_s} \nonumber\\
&+|{\bf w}^{*}_{n,m}{\bf H}_{n,m+\Delta m}[:,l']|^2{E_s}\nonumber\\&+\sum_{j=1}^{l'-1} |{\bf w}^{*}_{n,m}{\bf H}^{H}_{n,m+\Delta m}[:,j]|^2{\sigma^2_{e_{n,(m'+j)}^{(i)}}}\nonumber\\&+\sum_{k=l'+1}^{l'+l_{max}}|{\bf w}^{*}_{n,m}{\bf H}_{n,m+\Delta m}[:,k]|^2\sigma^2_{e_{n,(m'+k)}^{(i-1)}}+||{\bf w}^{*}_{n,m}||^2\sigma_{n}^2.\label{eqapp4}
\end{align}
{Assuming that the Doppler indices for different paths are approximated by the maximum Doppler index $\nu_{max}$. Thus, for each $(m+\Delta m)$-th sub-channel, the minimized MSE $\beta_{min}=|{\bf w}^{*}_{n,m}{\bf H}_{n,m+\Delta m}[:,l']|^2{E_s}+\sum_{j=1}^{l'-1} |{\bf w}^{*}_{n,m}{\bf H}^{H}_{n,m+\Delta m}[:,j]|^2{\sigma^2_{e_{n,(m'+j)}^{(i)}}}+\sum_{k=l'+1}^{l'+l_{max}}|{\bf w}^{*}_{n,m}{\bf H}_{n,m+\Delta m}[:,k]|^2\sigma^2_{e_{n,(m'+k)}^{(i-1)}}+||{\bf w}^{*}_{n,m}||^2\\\sigma_{n}^2$ remains unchanged, and $\Delta \beta$ is then}
\begin{align}
\Delta\beta &= \nonumber\beta-\beta_{min}\\&=\nonumber-{\bf w}^{*}_{n,m}{\bf H}_{n,m+\Delta m}[:,l']{E_s}-\bar{{\bf w}}^{*}_{n,m}\bar{{\bf H}}_{n,m+\Delta m}[:,l']{E_s}\\&=(-\mu_{n,m+\Delta m}-\bar{\mu}_{n,m+\Delta m})E_{s}.\label{eqapp5}
\end{align}
Let $\mu_{n,m} = {\bf w}^{*}_{n,m}{\bf H}_{n,m}[:,l'] = a{e^{j\theta}}$, which is the scalar coefficient at the $m$-th sub-channel. And due to maximum Doppler assumption, the $l'$th column at $(m+\Delta m)$-th sub-channel ${\bf H}_{n,m+\Delta m}[:,l']$ is phase rotated by
${e^{j\frac{2\pi}{MN}\Delta m\nu_{max}}}$. Thus, $\mu_{n,m+\Delta m} =\mu_{n,m}{e^{j\frac{2\pi}{MN}\Delta m\nu_{max}}}= a{e^{j\theta}}{e^{j\frac{2\pi}{MN}\Delta m\nu_{max}}}$, which yields
\begin{align}
\Delta\beta = 2acos(\theta+\frac{2\pi}{MN}\Delta m\nu_{max}).\label{eqapp6}
\end{align}
By re-arranging (\ref{eqapp6}), we will arrive at the relation in (\ref{44}).
\bibliographystyle{IEEEtran}
\bibliography{ref2,ref,references}

\end{document}